\documentclass[lettersize,onecolumn]{IEEEtran}
\IEEEoverridecommandlockouts
\usepackage{cite}
\usepackage{subfigure}
\usepackage{amsmath,amssymb,amsfonts}
\usepackage{algorithmic}
\usepackage{graphicx}
\usepackage{textcomp}
\usepackage{xcolor}
\usepackage{amsthm}
\theoremstyle{plain}
\newtheorem{theorem}{\textbf{Theorem}}
\newtheorem{lemma}{\textbf{Lemma}}

\newtheorem{definition}{Definition}

\newtheorem{proposition}{\textbf{Proposition}}
\newtheorem{remark}{\textbf{Remark}}
\newtheorem{corollary}{\textbf{Corollary}}
\def\BibTeX{{\rm B\kern-.05em{\sc i\kern-.025em b}\kern-.08em
    T\kern-.1667em\lower.7ex\hbox{E}\kern-.125emX}}
\begin{document}

\title{Generative Rate-Distortion-Perception Theory on Riemannian Manifolds}
\author{Tong Wu, Zhiyong Chen, \emph{Senior Member, IEEE}, Meixia Tao, \emph{Fellow, IEEE}, Dazhi He, \emph{Member, IEEE},\\ Feng Yang, and Wenjun Zhang, \emph{Fellow, IEEE}
\thanks{The authors are with the School of Information Science and Electronic Engineering, Shanghai Jiao Tong University, Shanghai 200240, China (Email: \{wu\_tong, zhiyongchen, mxtao, hedazhi, yangfeng, zhangwenjun\}@sjtu.edu.cn). \emph{Corresponding author: Zhiyong Chen.}}}

\maketitle
%\begin{abstract}
%Generative communication is rapidly developing, where a generative model is adopted at the receiver to generate high-quality data by randomly sampling from a conditional distribution controlled by the transmitted information. A fundamental theory for generative communication should therefore not only characterize the relation among rate, distortion and conditional-distribution based perception quality, but also capture the intrinsic geometric structure of real-world sources and the reliability of stochastic generation. 
%To this end, we consider an asymptotic operational generative rate-distortion-perception (GRDP) formulation with conditional-distribution preservation, where Riemannian manifolds are introduced to model the source space and an excess-distortion probability based on the intrinsic geodesic distance is adopted as the distortion measure to describe the reliability of stochastic generation.  We derive its single-letter characterization through an equivalent informational formulation and the corresponding geometric lower and upper bounds in high-reliability regime through Riemannian geometry, revealing how the intrinsic dimension and Riemannian curvature of the source manifold affect the required communication rate. 
%In the case of homogeneous Riemannian manifolds with uniformly distributed source, we prove that these bounds coincide, yielding a tight characterization. Finally, spherical manifolds are studied as a representative example of homogeneous manifolds to demonstrate the explicit computation of the proposed rate bounds.
%\end{abstract}

\begin{abstract}
Generative communication uses transmitted information
to control stochastic generation at the receiver.
A fundamental question is how much source-specific
information must be communicated so that outputs
generated at the receiver satisfy prescribed fidelity
and reliability requirements while preserving the
conditional source distribution.
We address this question for sources on compact
Riemannian manifolds by developing a generative
rate-distortion-perception (GRDP) framework without
common randomness. We formulate an asymptotic operational
GRDP function under a geodesic distortion radius and
a failure-probability constraint, and establish its
single-letter characterization through an equivalent
informational formulation. Furthermore, for the informational formulation with the high-reliability limit, we derive a geometric
lower bound determined by the largest source probability
mass of a region with bounded geodesic diameter,
together with an upper bound based on probabilistic
coverings of the source manifold. Small-radius analysis
of the lower bound reveals the roles of intrinsic
dimension and scalar curvature in the required
communication rate.
For uniformly distributed sources on compact homogeneous
Riemannian manifolds, we prove that the bounds coincide,
yielding an exact characterization of the limiting
GRDP function. We further specialize this result to unit spheres, obtaining an explicit expression
that relates the minimum communication rate to the
geodesic distortion radius and intrinsic dimension, revealing that for the $d$-dimensional unit sphere with $d\ge2$, increasing the communication rate by $d$ bits per source symbol can reduce the distortion radius by more than half.
\end{abstract}

\section{Introduction}
Recent advances in generative artificial intelligence have made it increasingly feasible to model natural data distributions and generate high-quality samples from learned models \cite{Score-DM,lipman2023flow,z-image,hunyuanvideo,wan2025}. These developments have motivated generative communication, in which transmitted information controls a stochastic generation process at the receiver \cite{GenCom1,GenCom2,GenCom3,GenCom4,GenCom5}. This paradigm shifts the emphasis from reconstructing a source representation through conventional decoding to conveying source-specific conditions that guide generation at the receiver. In this paradigm, the learned source prior captures statistical regularities shared across source realizations, while the transmitted message conveys information about the particular realization to be communicated. The receiver conditions the generative model on this message and samples an output using local randomness \cite{JSCGC}, as illustrated in Fig.~\ref{fig:system_model}. This approach has been explored for image transmission \cite{HiFiC,diffcom}, video delivery \cite{gao2025givic,mao2026generative,GVSC}, immersive media \cite{sehad2025generative,yang2026gensecpc}, and extended-reality services \cite{yang2024streamlined}.

These developments call for an information-theoretic framework that characterizes the fundamental limits of generative communication. A key question is \textbf{how much source-specific information must be communicated for outputs generated at the receiver to satisfy prescribed fidelity and reliability requirements while preserving the conditional source distribution}. Answering this question requires characterizing the minimum communication rate under joint fidelity, reliability, and conditional-distribution constraints.

The rate-distortion-perception (RDP) framework provides a natural starting point \cite{blau2019rethinking}, as it jointly characterizes communication rate, reconstruction fidelity, and distributional realism. It extends classical rate-distortion theory by introducing a perception constraint that measures the discrepancy between the source and reconstruction distributions. Subsequent work established coding theorems for the RDP function and clarified the role of randomness in encoding and decoding \cite{theis2021coding,chen2022rate}. In particular, under symbol-wise perception constraints, the operational RDP function admits a single-letter characterization without common randomness \cite{chen2022rate}. Under sequence-level marginal perception constraints, however, obtaining a general characterization without common randomness is more challenging. Existing studies have therefore exploited additional structure in the perception constraint, distortion measure, or source distribution.

One line of work considers the perfect-perception setting, in which the reconstruction is required to preserve the source distribution \cite{saldi2015output,chen2022rate,hamdi2024rate}. These studies establish single-letter characterizations under the corresponding operational criteria and examine the role of randomness. Another line exploits specific distortion and perception measures to obtain tractable characterizations. For example, \cite{qu2025rate} establishes a single-letter characterization under squared-error distortion and a squared Wasserstein-2 perception measure with limited common randomness, together with an explicit evaluation for Gaussian sources. In addition, \cite{serra2024copula} develops copula-based methods for estimating the perfect-perception RDP function of continuous sources and derives analytical lower bounds under mean squared error (MSE) distortion.

While these studies primarily constrain marginal distributions, \cite{salehkalaibar2024rdp} investigates a conditional-distribution-based perception measure that compares the source and reconstruction distributions conditioned on the encoder output. Building on the posterior reference map introduced in \cite{atif2024posterior}, it establishes a single-letter characterization under sequence-level conditional perception constraints without common randomness for discrete memoryless sources. For continuous memoryless sources, it further establishes a characterization in Euclidean spaces under squared-error distortion and a squared Wasserstein-2 perception measure, including the perfect-perception case. For Gaussian vector sources, the resulting characterization admits a reverse water-filling-type solution. These results provide a foundation for analyzing how transmitted information constrains the distribution of generated outputs.

Despite these advances, existing RDP characterizations do not directly address the joint role of intrinsic source geometry and generation reliability under conditional-distribution preservation. Extending these results to the setting considered here requires addressing two challenges.

First, many tractable RDP characterizations for continuous sources are developed in Euclidean spaces with squared-error distortion. Natural media sources, however, may exhibit low-dimensional geometric structure within a high-dimensional ambient space. The manifold hypothesis suggests that their distributions can concentrate on or near low-dimensional manifolds \cite{seung2000manifold,fefferman2016testing,farghly2025diffusion,roweis2000nonlinear}. For example, although facial images are represented by high-dimensional pixel vectors, much of their variation can be described by a small number of factors, such as pose, expression, and illumination \cite{seung2000manifold}. This motivates modeling structured sources on Riemannian manifolds and measuring fidelity using intrinsic geodesic distance. Although such sources can also be represented in an ambient Euclidean space, ambient distance need not capture their intrinsic geometry. A geometric characterization is therefore needed to make explicit how the intrinsic dimension and curvature of the source manifold affect the required communication rate.

Second, conditional-distribution consistency and generation reliability impose distinct requirements. In generative communication, the receiver samples an output from a conditional distribution specified by the received information. Even when this distribution matches the conditional source distribution exactly, an individual generated output need not be close to the actual source realization. Different outputs sampled under the same received information can therefore have substantially different distortions. Average-distortion criteria do not explicitly quantify the probability of exceeding a prescribed distortion threshold. The relevant question is thus how much communication is required to satisfy a prescribed geodesic distortion tolerance with a target failure probability while preserving the conditional source distribution.

\begin{figure*}
    \centering
    \includegraphics[width=1\textwidth]{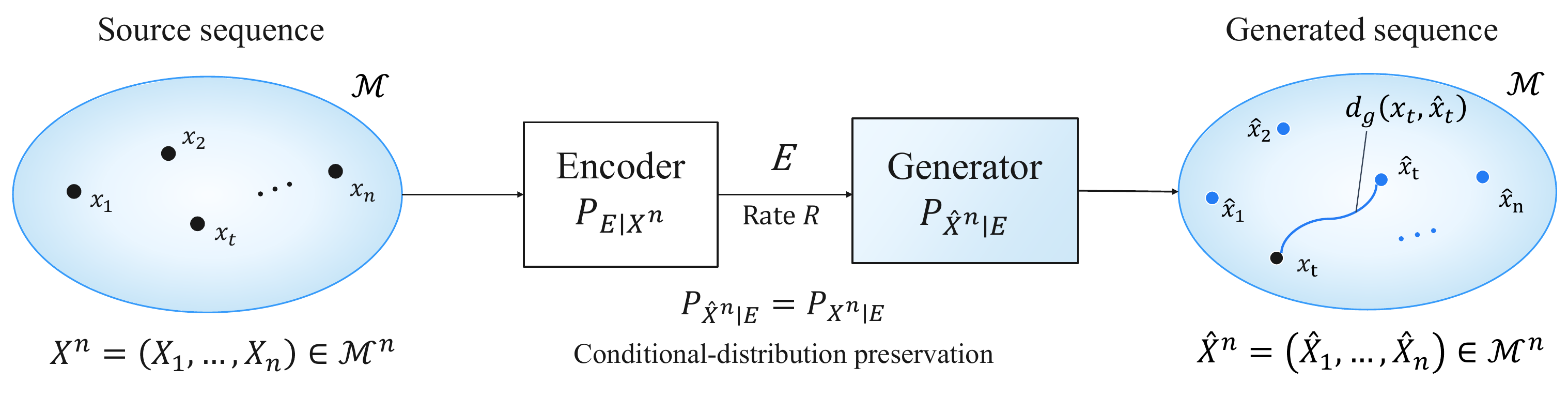}
    \caption{Generative communication model on a Riemannian
manifold under conditional-distribution preservation.
The black point on the right serves as the source
reference for illustrating the geodesic distortion
$d_g(\mathbf x_t,\hat{\mathbf x}_t)$.}
\label{fig:system_model}
\end{figure*}

Motivated by these considerations, we extend the Euclidean source model to sources taking values on Riemannian manifolds and measure fidelity using the geodesic distance induced by the Riemannian metric. To capture the reliability of stochastic generation, we adopt a geodesic excess-distortion criterion \cite{kostina2012fixed}, which quantifies the probability that the distortion between the source and the generated output exceeds a prescribed threshold. Together with conditional-distribution preservation, this criterion allows us to characterize the tradeoff among communication rate, geodesic distortion radius, and generation failure probability, thereby quantifying the fidelity and reliability achievable through communication-controlled generation.

The contributions of this paper can be summarized as follows:
%we formulate an asymptotically operational rate function without common randomness that explicitly accounts for describing both the intrinsic geometry of the source space and the reliability of generation with conditional-distribution preservation. 
% Furthermore, we derive a single-letter characterization of the proposed asymptotically operational rate function. Building upon this characterization, we further develop its geometric lower bounds through Riemannian geometric analysis, and an achievable upper bound by introducing the vertex-packing function, revealing the influence of geodesic feature of source data, such as intrinsic dimension and curvature on the RDP trade-off. Moreover, we consider a special case of homogeneous manifolds, proving that the upper and lower bounds coincide, yielding a tight characterization on homogeneous manifolds. Finally, we investigate the spherical manifold as a representative homogeneous manifold to demonstrate the computation and implication of the proposed theoretical bounds.

\begin{itemize}
   % \item We formulate an asymptotic operational generative RDP (GRDP) function for manifold-valued sources under prescribed geodesic distortion and failure-probability constraints, subject to conditional-distribution preservation without common randomness. This formulation incorporates both the intrinsic geometry of the source space and the reliability of stochastic generation. We also define a corresponding informational GRDP function and establish its equivalence to the operational function in Theorem \ref{theorem:coding_theorem}, yielding a single-letter characterization.
\item We formulate an asymptotic operational generative
rate-distortion-perception (GRDP) function on Riemannian
manifolds under exact conditional-distribution preservation
and without common randomness. The formulation incorporates
a prescribed geodesic distortion radius and an
excess-distortion probability to capture the fidelity
and reliability of stochastic generation. In Theorem~\ref{theorem:coding_theorem}, we establish
an equivalent single-letter informational characterization,
thereby identifying the minimum asymptotic communication
rate required to control generation under these constraints.

%\item Using Riemannian geometric analysis, we derive geometric lower and upper bounds on the single-letter characterization in the high-reliability regime in Theorems \ref{thm_generative_rate_distortion_converse} and \ref{thm_general_alpha_achievability}, respectively. These bounds reveal how the rate-distortion trade-off depends on the intrinsic geometric properties of the source manifold, including its intrinsic dimension and Riemannian curvature.

\item We derive geometric lower and upper bounds on the
informational GRDP function in the high-reliability regime in
Theorem~\ref{thm_generative_rate_distortion_converse}
and Theorem~\ref{thm_general_alpha_achievability}, respectively.
The lower bound relates the necessary communication rate
to the maximum source probability contained in a region
of bounded geodesic diameter. The upper
bound provides an achievable construction based on
probabilistic coverings by such regions, with overlapping
regions offering flexibility in representing the source.
A small-radius analysis further makes explicit how the
intrinsic dimension and scalar curvature contribute to
the rate lower bound.

  %  \item For uniformly distributed sources on homogeneous Riemannian manifolds, we establish in Theorem \ref{thm:homogeneous_upper_bound} that the lower and upper bounds coincide in the high-reliability regime, yielding a tight characterization of the GRDP function. We further specialize this characterization to spherical manifolds, obtaining an explicit evaluation of the communication rate required to satisfy prescribed geodesic distortion and failure-probability constraints.

\item We establish an exact characterization of the informational GRDP
function in the high-reliability limit for uniformly
distributed sources on compact homogeneous Riemannian
manifolds in Theorem~\ref{thm:homogeneous_upper_bound}.
Specifically, we show that the geometric lower bound
is achievable by randomizing over isometric images
of admissible regions. The minimum communication rate
is thus determined by the largest source probability 
mass of a region satisfying the prescribed geodesic
diameter constraint. We further specialize this result
to unit spheres, obtaining an explicit expression
that relates the required rate to the geodesic distortion
radius and intrinsic dimension, revealing that for the $d$-dimensional unit sphere with $d\ge2$, increasing the communication rate by $d$ bits per source symbol can reduce the distortion radius by more than half.
\end{itemize}

The remainder of this paper is organized as follows. Section II introduces the source and system models, motivates the distortion and perception criteria, and defines the operational and informational GRDP functions under conditional-distribution preservation. Section III establishes the equivalence of these functions, yielding a single-letter characterization, and derives lower and upper bounds that reveal the role of intrinsic source geometry. Section IV establishes the tightness of these bounds for homogeneous manifolds and illustrates their evaluation for spherical sources. Section V concludes the paper.

\section{Problem Statement}
\subsection{Geometric Source Model}
% To characterize the intrinsic geometry of the source space, we model the support of the source distribution as a compact Riemannian manifold $\mathcal M$ embedded in an ambient Euclidean space $\mathbb R^m
% To mathematically analyze the intrinsic geometric properties of natural data, we model the intrinsic data space as a Riemannian manifold, given that natural data $\mathbf{x}$ resides on a low-dimensional manifold $\mathcal{M}$ embedded in the high-dimensional ambient space $\mathbb{R}^m$. The intrinsic dimension of the manifold, denoted by $d$, satisfies $d < m$.
As shown in Fig.~\ref{fig:system_model}, we consider to characterize the fundamental limits of generative communications, which transmits the continuous-valued natural data from the real world. As suggested in \cite{seung2000manifold,fefferman2016testing,farghly2025diffusion,roweis2000nonlinear}, the real-world data may concentrate on low-dimensional manifolds rather than occupying the entire ambient Euclidean space. Therefore in this case, it is necessary to explicitly account for the intrinsic geometry of the source space and investigate its impact on communication rate. To this end, we model the support of the source distribution as a Riemannian manifold $\mathcal{M}$ with intrinsic dimension $d$,  embedded in $\mathbb{R}^{m}$, with $d<m$.

The manifold $\mathcal{M}$ is then equipped with the necessary geometric structures.
% Therefore, all the measure-theoretic and information-theoretic we will discuss in the following sections are defined intrinsically on the manifold $\mathcal{M}$ rather than the ambient space $\mathbb{R}^m$.
% Therefore, we first equip the manifold $\mathcal{M}$ with a fundamental geometric structure to mathematically describe this manifold. Specifically, we define the following key components:
\begin{itemize}
    \item \textbf{Tangent Space} $T_\mathbf{x}\mathcal{M}$: Denotes the tangent space of $\mathcal{M}$ at a point $\mathbf{x}$, which is a $d$-dimensional linear vector space.
    \item \textbf{Riemannian Metric} $g$: Let $g$ be a Riemannian metric on $\mathcal M$. At each point $\mathbf x\in\mathcal M$, $g_{\mathbf x}$ is a positive-definite symmetric bilinear form on the tangent space $T_{\mathbf x}\mathcal M$.
    \item \textbf{Scalar Curvature} $S(\mathbf{x})$: Denotes the scalar curvature of the manifold at point $\mathbf{x}$ induced by $g$, which quantifies the infinitesimal bending of the local space at that point.
    \item \textbf{Volume measure} $dV_g$: Induced by the Riemannian metric $g$ with respect to the standard Euclidean coordinate volume element $d\mathbf{x} = d\mathbf{x}^1 \dots d\mathbf{x}^d$ is given by:
    \begin{equation}
    dV_g = \sqrt{\det g}   d\mathbf{x},
    \end{equation} 
    \item \textbf{Volume of a region} $V_g(O)$: For any Borel measurable subset $O \subseteq \mathcal{M}$, its Riemannian volume is defined as the Lebesgue integral of the volume measure over the target space $O$:
    \begin{equation}
    V_g(O) = \int_{O} dV_g.
    \end{equation} 
    \item \textbf{Curves Space} $\Omega(\mathbf{x}_1,\mathbf{x}_2)$: Denotes the space of all piecewise smooth parameterized curves on $\mathcal{M}$ connecting two points $\mathbf{x}_1$ and $\mathbf{x}_2$ on the manifold.
    \item \textbf{Riemannian length} $L(\gamma)$: is defined as the integral of the norm of its tangent vector $\dot{\gamma}(t) = \frac{d\gamma}{dt} \in T_{\gamma(t)}\mathcal{M}$ under the metric $g$ for any curve $\gamma \in \Omega(\mathbf{x}_1,\mathbf{x}_2)$ in this space,
    \begin{equation}
    L(\gamma) = \int_0^1 \sqrt{g_{\gamma(t)}(\dot{\gamma}(t), \dot{\gamma}(t))}   dt.
    \end{equation}
    \item \textbf{Geodesic Distance} $d_g(\mathbf{x}_1, \mathbf{x}_2)$: is defined as the infimum of the Riemannian lengths of all admissible curves in $\Omega(\mathbf{x}_1,\mathbf{x}_2)$ connecting two points $\mathbf{x}_1$ and $\mathbf{x}_2$ on the Riemannian manifold:
  \begin{equation}
      d_g(\mathbf{x}_1, \mathbf{x}_2) = \inf_{\gamma \in \Omega(\mathbf{x}_1,\mathbf{x}_2)} L(\gamma).
  \end{equation}
  \item \textbf{Geodesic Ball} $B_r(\mathbf{x})$: Is formally defined as the geodesic ball centered at a given point $\mathbf{x}$ with a radius $r>0$:
  \begin{equation}
  B_r(\mathbf{x}) = \{ \hat{\mathbf{x}} \in \mathcal{M} \mid d_g(\mathbf{x}, \hat{\mathbf{x}}) \le r \}.
    \end{equation}

   % \item \textbf{$\mathcal{V}_{\mathbf{X}}(r)$ Inverse Function}  $\mathcal{V}_{\mathbf{X}}^{-1}(\cdot)$: is defined as the inverse function of $\mathcal{V}_{\mathbf{X}}(r)$, which is global existence and uniqueness.
\end{itemize}

Beside the basic structure, the manifold $\mathcal{M}$ should be connected and compact, meaning that $\forall \mathbf{x}_1,\mathbf{x}_2 \in \mathcal M$, the geodesic distance is attained and
\begin{align}
    d_g(\mathbf{x}_1,\mathbf{x}_2)<\infty, \quad V_g(\mathcal M)<\infty \quad \forall\mathbf{x}_1,\mathbf{x}_2 \in \mathcal M
\end{align}
The compactness assumption is motivated by the bounded range of practical natural data, while the connectedness can be imposed without loss of generality by restricting attention to a connected component of the data space.

With the mathematical definition, we can now define the probability density of natural data ${\mathbf{X}}$ by the Radon-Nikodym derivative as 
\begin{align}
    f(\mathbf x):=\frac{dP_{{\mathbf{X}}}}{dV_g}(\mathbf x)
\end{align}
with its Lebesgue integral $\int_{\mathcal M}f(\mathbf x) dV_g(\mathbf x)=1$, where $P_{{\mathbf{X}}}$ is a probability measure on $\mathcal M$ that is absolutely continuous with respect to the Riemannian volume measure $V_g$.

In conclusion, in this paper, we consider the case that the support of source $\mathbf{X}$ is a connected and compact Riemannian manifold and its density is therefore naturally defined on the manifold by Radon-Nikodym derivative. 
Moreover, it is obvious that the quantities related to the probability should be defined on the manifold $\mathcal{M}$ by Lebesgue integral with respect to the Riemannian volume measure $V_g$. For example, the differential entropy of ${\mathbf{X}}$ is defined as
\begin{align}
    h({\mathbf{X}}) = -\int_{\mathcal{M}} f(\mathbf{x}) \log_2 f(\mathbf{x}) dV_g.
\end{align}
%Therefore, all subsequent probability and  information-theoretic quantities 
\begin{remark}
    This Riemannian manifold model directly characterizes the geometry of the source support, rather than treating the source merely as a random variable in the ambient Euclidean space. Moreover, it generalizes the conventional Euclidean source model, since the Euclidean space is itself a special case of a Riemannian manifold.
\end{remark}

\subsection{Generative Communications Model and Problem Definition}
Let $\{\mathbf{X}_t\}_{t=1}^{\infty}$ denote a stationary and memoryless source with marginal distribution $P_{\mathbf{X}}$ on a connected and compact manifold $\mathcal{M}$. A stochastic encoder maps a length-$n$ source sequence $\mathbf{X}^n$ to a message $E$ in a binary prefix code $\mathcal{E}$ through the conditional distribution $P_{E|\mathbf{X}^n}$. At the receiver, a generator samples a length-$n$ sequence $\hat{\mathbf{X}}^n$ according to the conditional distribution $P_{\hat{\mathbf{X}}^n|E}$.

We next analyze the characteristics of generative communication to identify the requirements that the theoretical framework should address. First, the generator is modeled to approximate the conditional source distribution given the transmitted information. The theoretical framework should therefore incorporate a conditional-distribution constraint. Under suitable regularity conditions, modern conditional generative models, such as conditional diffusion models, provide theoretical guarantees that the learned conditional distribution can approximate the target conditional distribution arbitrarily well as the amount of training data and the model capacity increase, including in settings involving Riemannian manifolds \cite{tang2025conditional,ishtiaque2026universality}. This approximation capability is independent of the communication constraints. Accordingly, we consider the idealized case in which the generator perfectly learns the conditional source distribution, meaning that
\begin{align}
P_{\hat{\mathbf{X}}^n|E}=P_{{\mathbf{X}}^n|E} \quad \text{almost surely (a.s.)}.
\end{align}

Second, conditional-distribution preservation also implies that both the source sample and the generated sample lie on the same low-dimensional manifold. We therefore measure fidelity using the geodesic distance $d_g(\mathbf{x},\hat{\mathbf{x}})$ induced by the Riemannian metric. Moreover, generative communication differs from conventional deterministic reconstruction in that the received message specifies a distribution rather than a unique reconstruction. The same message may lead to different generated samples, and the resulting distortion depends on both the source realization and the private randomness at the receiver. To characterize the reliability of stochastic generation controlled by the transmitted information, we define the geodesic excess-distortion probability as
\begin{align}
P(r):=\Pr\{d_g(\mathbf{X},\hat{\mathbf{X}})\ge r\},
\end{align}
which quantifies the probability that the geodesic distortion is at least a prescribed radius $r$.

These observations motivate the following operational formulation of the fundamental limits of generative communication.
\begin{definition}\label{def:achieveable_rate}
For a radius $r>0$ and a failure probability $0\le\epsilon<1$, a rate $R$ is said to be asymptotically achievable under conditional-distribution preservation if there exists a sequence of stochastic block codes indexed by $n$, each consisting of a stochastic encoder $P_{E|\mathbf{X}^n}$ and a stochastic decoder $P_{\hat{\mathbf{X}}^n|E}$, satisfying
\begin{align} 
&\limsup_{n\to\infty}\frac{1}{n}\mathbb E[\ell(E)]\le R. \\
&\limsup_{n\to\infty}\frac{1}{n}\sum_{t=1}^{n}\Pr\{d_g({\mathbf{X}}_t,\hat{{\mathbf{X}}}_t)\ge r\}\le\epsilon. \\
&P_{\hat{{\mathbf{X}}}^n|E}=P_{{\mathbf{X}}^n|E}\qquad P_{E}\text{-a.s.} 
\end{align}
where $\ell(E)$ denotes the length of $E$. The asymptotic operational GRDP function is defined as
\begin{align} R_{\mathrm C}(r,\epsilon):=\inf\{R\ge0:R\text{ is achievable for }(r,\epsilon)\}. \end{align}
\end{definition}

To characterize $R_{\mathrm C}(r,\epsilon)$, we introduce the corresponding informational GRDP function.
\begin{definition}
\label{def_generative_rate_distortion}
For a radius $r>0$ and a failure probability $0\le\epsilon<1$, the informational GRDP function is defined as
\begin{align}
R(r,\epsilon):=&\inf_{P_{\hat U|{\mathbf{X}}}, P_{\hat{\mathbf{X}}|\hat U}}I({\mathbf{X}};\hat U)\\
\text{s.t.} \quad &{\mathbf{X}}-\hat U-\hat {\mathbf{X}},\\
&\Pr\{d_g({\mathbf{X}},\hat {\mathbf{X}})\ge r\}\le\epsilon,\\
&P_{\hat {\mathbf{X}}|\hat U}=P_{{\mathbf{X}}|\hat U}\qquad P_{\hat U}\text{-a.s.}.
\end{align}
where the auxiliary random variable $\hat{U}$ is a representation of $\mathbf{X}$ and serves as the single-letter counterpart of the message $E$ in Definition \ref{def:achieveable_rate}.
\end{definition}

\section{Main Results}
In this section, we first establish the equivalence between the asymptotic operational and informational GRDP functions, yielding a single-letter characterization. We then derive lower and upper bounds on this characterization for a prescribed geodesic distortion radius $r$ in the limit of vanishing failure probability.

\begin{theorem}\label{theorem:coding_theorem}
For $r>0$ and $0<\epsilon<1$,
% Assume $(\mathcal{M}, g)$ be a compact and connected manifold, for $r>0$ and $0<\epsilon<1$,
\begin{align} R_{\mathrm C}(r,\epsilon) = R(r,\epsilon). \end{align}
\end{theorem}

\begin{proof}
The complete proof is provided in Appendix \ref{app:coding_theorem}. The converse follows from a standard time-sharing argument. For achievability, we first quantize the compact manifold into a sufficiently fine finite partition and introduce a robust finite-alphabet threshold loss that asymptotically preserves the original distortion constraint. The support lemma is then used to replace the possibly nonfinite auxiliary variable by a finite one without changing the relevant mutual information and expected loss. A finite-alphabet posterior-reference map lemma is subsequently applied to construct a block code satisfying the conditional-distribution constraint, and the finite-alphabet reconstruction is finally lifted back to the manifold. 
\end{proof}

%--------------------- 下界证明 --------------------------------- %
Theorem~\ref{theorem:coding_theorem} establishes the operational significance of $R(r,\epsilon)$: every rate strictly above it is asymptotically achievable by the generative communications. With the theorem, we concentrate on the single-letter characterization with $\epsilon\downarrow0$ and define
\begin{align}
R(r)=\lim_{\epsilon \downarrow0} R(r, \epsilon).
\end{align}

We next derive a lower bound on $R(r)$. To this end, we first establish the following proposition.
\begin{proposition}\label{pro:constraint_covering}
Let $\mathcal{P}(\mathcal{M})$ denote the space of Borel probability measures on $\mathcal{M}$, and let $\mathcal B(\mathcal M)$ denote
the Borel $\sigma$-algebra of $\mathcal M$. For a fixed $r>0$,
\begin{align}
\liminf_{\tau\downarrow0}\psi_r(\tau)\ge-\log_2\beta_r,
\end{align}
where, for $\tau>0$,
\begin{align}
&\psi_r(\tau):=\inf_{Q \in Q_r(\tau)}D_{KL}(Q\|P_{\mathbf X}).\\
&Q_r(\tau)=\left\{Q\in\mathcal P(\mathcal M): \int_{\mathcal M}\int_{\mathcal M}\mathbf 1 \left\{d_g(\mathbf x, \hat{\mathbf{x}})\ge r\right\}Q(d\mathbf x)Q(d \hat{\mathbf{x}})\le\tau\right\}\\
&\beta_r:=\sup_{\substack{A\in\mathcal B(\mathcal M): \operatorname{diam}_g(A)\le r}}P_{\mathbf X}(A),\\
&\operatorname{diam}_g(A):=\max_{\mathbf{x},\hat{\mathbf{x}}\in\mathcal A}d_g(\mathbf \mathbf{x},\hat{\mathbf x}).
\end{align}
\end{proposition}
\begin{proof}
    The proof is detailed in Appendix \ref{app:constraint_covering}.
\end{proof}
As shown in Fig.\ref{fig:cost}, the proposition can be illustructed in geodesic that if a distribution $Q$ makes the probability that two independent samples drawn from $Q$ are separated by a geodesic distance of at least $r$ tends to zero, then the probability mass of $Q$ must asymptotically concentrate on a region whose geodesic diameter is no greater than $r$. Meanwhile, restricting $P_{\mathbf X}$ to such a region incurs a relative-entropy cost of at least $-\log\beta_r$.

 %Equipped with the geometric and probabilistic definitions, we now state our main theorem. This theorem establishes a lower bound on the information capacity required for a generative model to achieve a specific level of reconstruction fidelity on a curved manifold.
\begin{figure}
    \centering
    \includegraphics[width=0.5\textwidth]{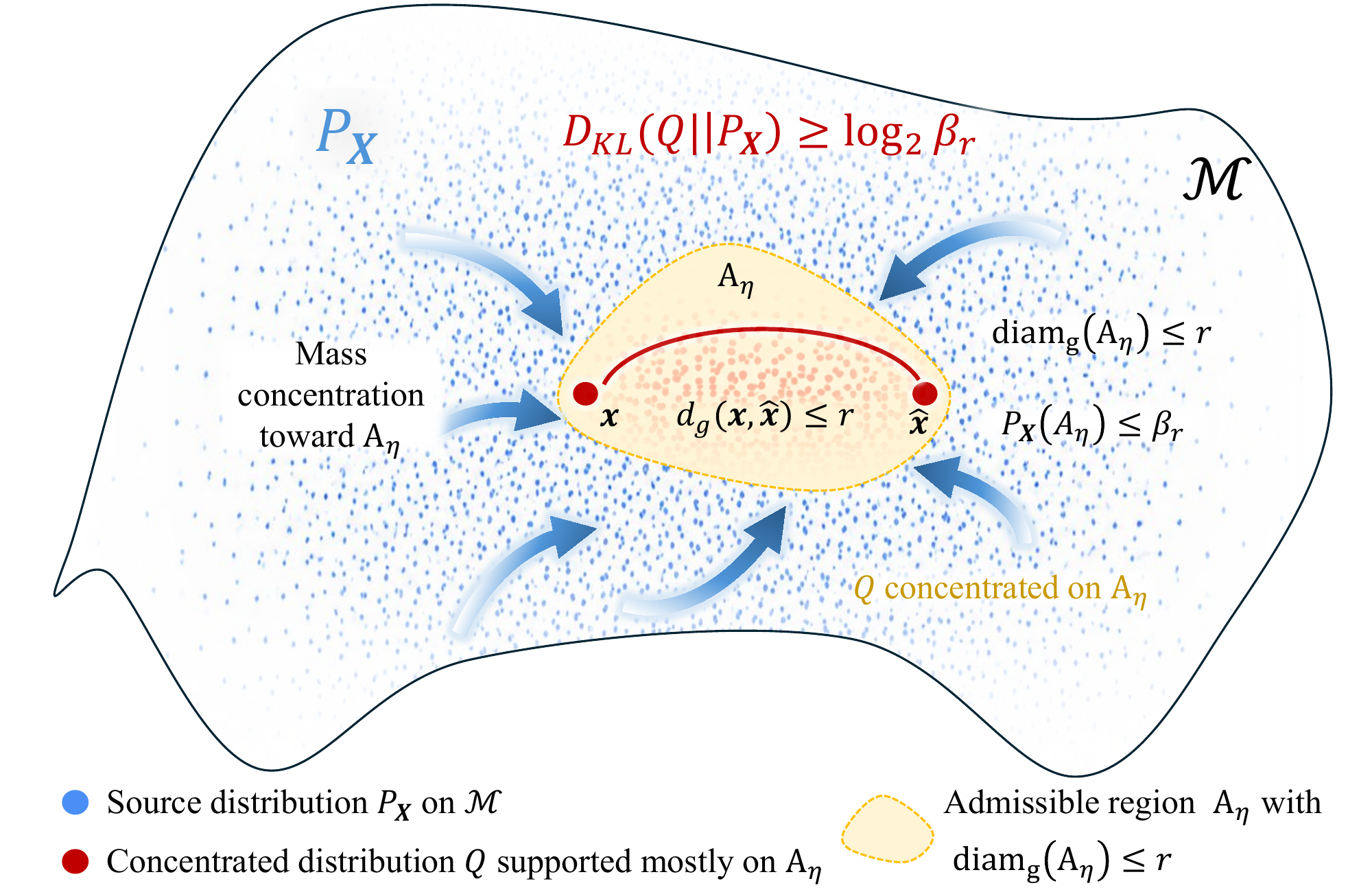}
    \caption{Geometric interpretation of Proposition \ref{pro:constraint_covering}}
    \label{fig:cost}
\end{figure}

 With the lemma, we can derive the lower bound of $R(r)$.
 \begin{theorem}[Lower Bound]
    \label{thm_generative_rate_distortion_converse}
For $r>0$, 
    \begin{align}
        R(r)\ge -\log_2\beta_r.
    \end{align}
\end{theorem}
\begin{proof}
    The proof is detailed in Appendix \ref{app:generative_rate_distortion_converse}. %Briefly, we decompose the reconstruction uncertainty according to whether the geodesic distortion is below or above the prescribed radius. The maximum-entropy principle then bounds the conditional entropy in terms of the volume of the corresponding geodesic ball, while the contribution of the distortion-violation event vanishes as $\epsilon$ tends to zero. Finally, the marginal-preservation constraint together with the data-processing inequality converts this entropy bound into the desired lower bound on the required rate.
\end{proof}

Although the lower bound above provides a general characterization, its geometric implications remain somewhat implicit. To more explicitly describe the influence of the intrinsic geometry of source, we next relax the lower bound by enclosing the admissible region within a larger geodesic ball and then employ the small-radius Taylor expansion of the geodesic-ball volume.

For any Borel set $A\subseteq\mathcal M$ satisfying
\begin{align}
\operatorname{diam}_g(A)\le r,
\end{align}
we have
\begin{align}
A\subseteq B_r(\mathbf x),
\end{align}
where $\mathbf{x}\in A$. Therefore,
\begin{align}
P_{\mathbf X}(A)\le P_{\mathbf X}\left(B_r(\mathbf x)\right)\le\sup_{\mathbf{x}'\in\mathcal M}P_{\mathbf X}\left(B_r(\mathbf{x}')\right).
\end{align}
Let $f_{max}=||f_(\mathbf x)||_{\infty}$. Taking the supremum over all Borel sets $A$ satisfying $\operatorname{diam}g(A)\le r$ gives
\begin{align}
\beta_r\le\sup_{\mathbf x\in\mathcal M}P_{\mathbf X}\left(B_r(\mathbf x)\right).
\end{align}
Consequently, we obtain the relaxed ball-based lower bound
\begin{align}
R(r)&\ge-\log_2\sup_{\mathbf x\in\mathcal M}P_{\mathbf X}\left(B_r(\mathbf x)\right)
=-\log_2{\sup_{\mathbf{x}\in \mathcal{M}}\int_{B_r(\mathbf x)}f(\mathbf x)V_g(d{\mathbf x})}
\ge -\log_2 \left[f_{max} \sup_{\mathbf{x}\in \mathcal{M}}{V_g\left(B_r(\mathbf x)\right)}\right].
\end{align}
Then, when $r$ is small enough, the volume of a geodesic ball admits the expansion
\begin{align}
V_g(B_r(\mathbf x))=\kappa_d r^d\left(1-\frac{\operatorname{Scal}_g(\mathbf x)}{6(d+2)}r^2+O(r^4)\right),
\end{align}
where $\operatorname{Scal}_g(\mathbf x)$ is the scalar curvature at the point $\mathbf{x}$, and 
\begin{align}
\kappa_d=\frac{\pi^{d/2}}{\Gamma\left(\frac d2+1\right)}.
\end{align}
Therefore, 
\begin{align}\label{eq:low_bound_rate}
R(r)\ge-\log_2\left\{f_{max}\kappa_d r^d\left[1-\frac{S_{min}}{6(d+2)}r^2+O(r^4)\right]\right\}.
\end{align}
% Now, if we concentrate on the leading term with respect to $r^d$, we can get 
% \begin{align}\label{eq:leading_term}
% r \gtrsim C\ 2^{-\frac{1}{d}R(r)},
% \end{align}
% where 
% \begin{align}
%     C=(f_{max}\kappa_d)^{-\frac{1}{d}},
% \end{align}
% and $\gtrsim$ represents 
% \begin{align}
%     \liminf_{r\downarrow0} \frac{r}{C2^{-\frac{1}{d}R(r)}}\ge 1
% \end{align}
% Furthermore, if the first curvature correction is further retained, we have 
% \begin{align}\label{eq:curvature_term}
% r \ge C\ 2^{-\frac{1}{d}R(r)}\left[1+\frac{S_{min}}{6d(d+2)}r^2+O(r^4)\right],
% \end{align}
% Equation \eqref{eq:leading_term} reveals the fundamental influence of intrinsic geometry on the required rate that a larger intrinsic dimension $d$ weakens the exponential decay of the lower bound, indicating a higher rate requirement for achieving the same geometric precision in higher-dimensional manifolds. Meanwhile, the curvature correction in \eqref{eq:curvature_term} characterizes that positive scalar curvature reduces the volume of small geodesic balls, thereby increasing the required rate, whereas negative scalar curvature enlarges local volumes and reduces the required rate.

\begin{remark}[Small-radius behavior]
The lower bound on rate in \eqref{eq:low_bound_rate} reveals rate-radius relationship when $r$ is small enough. In particular, the leading-order
bound is
\begin{align}\label{eq:leading_term}
r\gtrsim C\,2^{-R(r)/d},
\qquad
C=(f_{\max}\kappa_d)^{-1/d},
\end{align}
where $\gtrsim$ means that
\begin{align}
\liminf_{r\downarrow0}
\frac{r}{C\,2^{-R(r)/d}}\ge1.
\end{align}

This expression identifies the intrinsic dimension $d$
as the coefficient of the logarithmic growth of the rate
lower bound with inverse radius. Equivalently, $d$
governs the exponential dependence of the radius lower
bound on rate. The second-order term further captures
the contribution of scalar curvature to this bound.
\end{remark}

We next derive an upper bound on the rate function to
investigate the tightness of the preceding lower bound.
To state this upper bound, we first introduce the
following definitions.

\begin{definition}[Admissible Regions]
\label{def_admissible_regions_coverage_functions}
For a prescribed radius $r>0$, define the family of admissible regions as
\begin{align}
\mathcal A_r:=\left\{A\in\mathcal B(\mathcal M):A\neq\varnothing,\;\operatorname{diam}_g(A)\le r\right\}.
\end{align}
\end{definition}

\begin{definition}[Coverage Functions]
Let $K\in\mathbb N$ be a positive integer, and let $A_1,\ldots,A_K\in\mathcal A_r$ be admissible regions. Let $J\in\{1,\ldots,K\}$ be a random index taking values in $\{1,\ldots,K\}$, with design weights $q_1,\ldots,q_K$
satisfying 
\begin{align}
\Pr(J=k)=q_k,\qquad q_k\ge 0,\qquad \sum_{k=1}^Kq_k=1.
\end{align}

For each $A\in\mathcal B(\mathcal M)$, let $\mathbf1_A$
denote its indicator function, defined by
\begin{align}
\mathbf 1_A(\mathbf x)=
\begin{cases}
1, & \mathbf x\in A,\\
0, & \mathbf x\notin A.
\end{cases}
\end{align}

The coverage function induced by
$\{A_k,q_k\}_{k=1}^K$ is the measurable function
$\alpha:\mathcal M\to[0,1]$ defined by
\begin{align}
\alpha(\mathbf x):=\Pr(\mathbf x\in A_J)=\sum_{k=1}^K q_k\mathbf 1_{A_k}(\mathbf x).
\end{align}
\end{definition}
\begin{definition}[Vertex-Packing Function Class]
Define $\operatorname{VP}_r$ as the class of all coverage
functions generated by finite mixtures of admissible
regions:
\begin{align}
\operatorname{VP}_r:=\left\{\alpha:
\begin{array}{l}\displaystyle\alpha(\mathbf x)=\sum_{k=1}^Kq_k\mathbf 1_{A_k}(\mathbf x)\text{ for some }K\in\mathbb N,\\
A_k\in\mathcal A_r,\;q_k\ge0,\;\displaystyle\sum_{k=1}^Kq_k=1
\end{array}
\right\}.
\end{align}

\end{definition}
\begin{definition}[$L^1(P_{{\mathbf{X}}})$-closure of coverage functions]
\label{def_coverage_function_closure}

Let $\overline{\operatorname{VP}}_r$ denote the closure
of $\operatorname{VP}_r$ in $L^1(P_{\mathbf X})$.
Specifically, a measurable function
$\alpha:\mathcal M\to[0,1]$ belongs to
$\overline{\operatorname{VP}}_r$ if and only if there
exists a sequence
$\{\alpha_n\}_{n\ge1}\subseteq\operatorname{VP}_r$
such that
\begin{align}
\lim_{n\to\infty}\mathbb E_{{\mathbf{X}}}\left[\left|\alpha_n({\mathbf{X}})-\alpha({\mathbf{X}})\right|\right]=0.
\end{align}
\end{definition}

% For $\alpha:\mathcal M\to[0,1]$, adopt the convention
% \begin{align}
% -\log_2 0:=+\infty.
% \end{align}
% Define the $\alpha$-coverage achievable functional by
% \begin{align}
% C_\alpha(r)
% :=
% \inf_{\alpha\in\operatorname{VP}_r}
% \mathbb E_{{\mathbf{X}}}
% \left[
% -\log_2\alpha({\mathbf{X}})
% \right].
% \end{align}
With these definitions, we can now state the general upper bound on $R(r)$.

\begin{theorem}[Upper Bound]
\label{thm_general_alpha_achievability}
Let $(\mathcal M,g)$ be a compact, connected Riemannian manifold. Then, for every $r>0$, the rate function satisfies
\begin{align}
R(r)\le\min_{\alpha\in\overline{\operatorname{VP}}_r}\mathbb E_{{\mathbf{X}}}\left[-\log_2\alpha({\mathbf{X}})\right].
\end{align}
\end{theorem}
\begin{proof}
The proof is detailed in Appendix \ref{app:thm_general_alpha_achievability}. We outline its two main steps below. First, we show that
\begin{align}
R(r)\le \inf_{\alpha\in\operatorname{VP}_r}\mathbb E_{\mathbf X}\left[-\log_2\alpha(\mathbf X)\right].
\end{align}
Starting from a finite coverage function $\alpha$ induced by sets of geodesic diameter at most $r$, we explicitly construct the auxiliary variable $\hat U$ by selecting a covering set containing the source realization, with probability proportional to its covering weight. The reconstruction is then generated from the corresponding posterior distribution. This construction ensures conditional-distribution-preserving constraint and satisfies the distortion constraint, since, conditional on $\hat U$, both $\mathbf X$ and $\hat{\mathbf X}$ lie in the same admissible region. A direct calculation of mutual information shows that
the achieved rate equals the logarithmic coverage cost
minus a nonnegative relative-entropy term, and is
therefore bounded above by
$\mathbb E_{\mathbf X}[-\log_2\alpha(\mathbf X)]$. Taking the infimum over all admissible coverage functions gives the first inequality.

Second, we show that
\begin{align}
\inf_{\alpha\in\operatorname{VP}_r}\mathbb E_{\mathbf X}\left[-\log_2\alpha(\mathbf X)\right]
=
\min_{\alpha\in\overline{\operatorname{VP}}_r}\mathbb E_{\mathbf X}\left[-\log_2\alpha(\mathbf X)\right].
\end{align}

\begin{figure}
    \centering
    \includegraphics[width=0.5\textwidth]{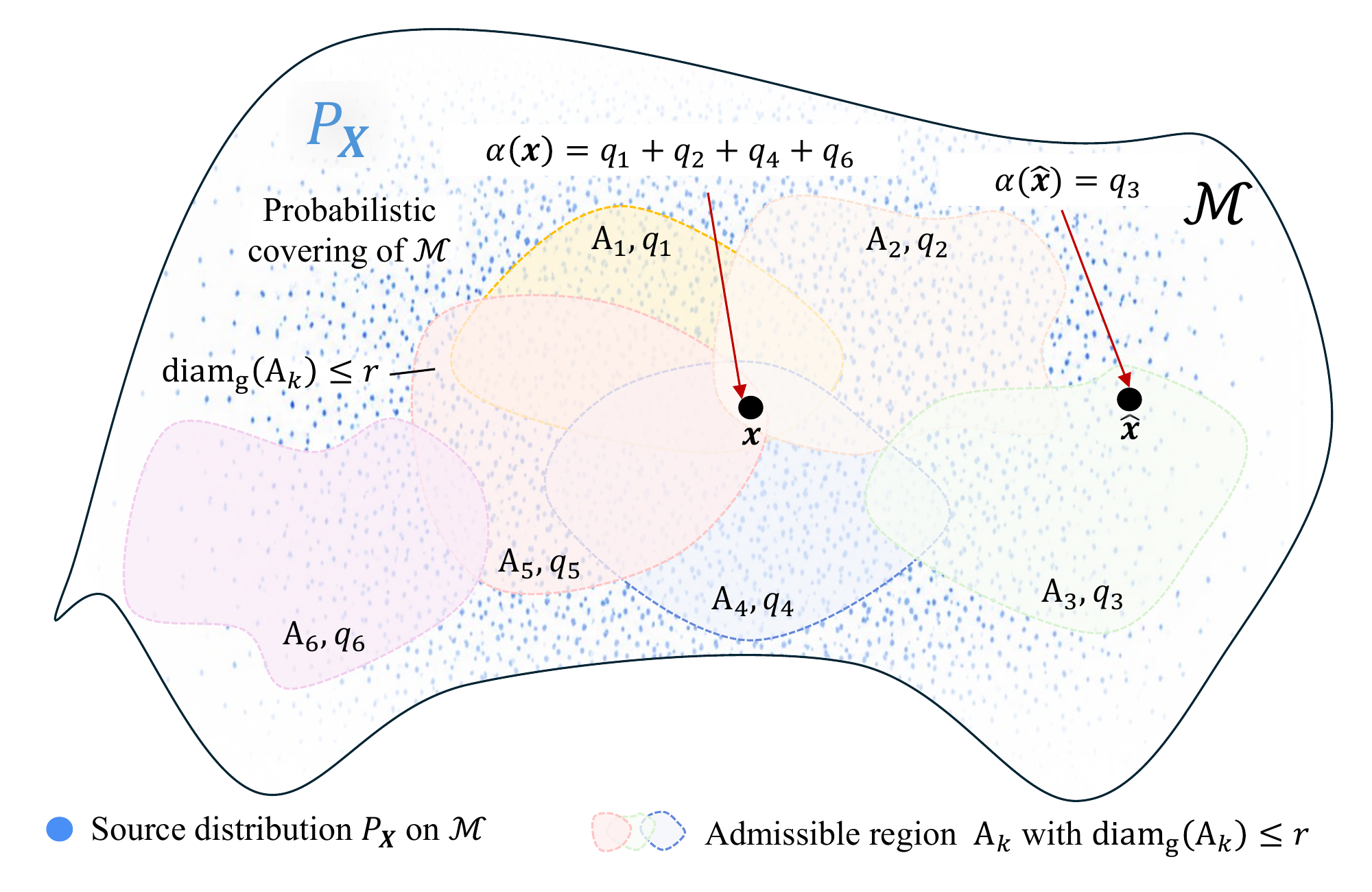}
    \caption{Geometric interpretation of Theorem \ref{thm_general_alpha_achievability}}
    \label{fig:upper_bound}
\end{figure}

We first show that taking the $L^1(P_{\mathbf X})$-closure
does not change the infimum. For a fixed positive
regularization parameter, we construct an approximating
sequence that is uniformly bounded away from zero,
allowing the logarithmic costs to converge. We then
let the regularization parameter tend to zero. We then use weak compactness of $\overline{\operatorname{VP}}_r$ in $L^1(P_{\mathbf X})$ together with the weak lower semicontinuity of the logarithmic cost functional to show that the infimum over the closure is attained.
\end{proof}

As illustrated in Fig.~\ref{fig:upper_bound},
Theorem~\ref{thm_general_alpha_achievability} admits
a probabilistic covering interpretation. Since admissible
regions may overlap, a point $\mathbf x\in\mathcal M$
may belong to several such regions. Its coverage value
is the sum of the design weights of all regions
containing it.

The theorem provides an achievable upper bound by
minimizing the expected negative logarithm of the
coverage value. The objective penalizes low coverage
more heavily in regions carrying greater source
probability mass, subject to the geometric constraints.

\begin{corollary}
\label{Cor_general_rate_distortion_bounds}
Combining the lower bound with the coverage-based upper
bound yields
\begin{align}
-\log_2 \beta_r \le R(r)\le\min_{\alpha\in\overline{\operatorname{VP}}_r}\mathbb E_{{\mathbf{X}}}\left[-\log_2\alpha({\mathbf{X}})\right].
\end{align}
\end{corollary}
\begin{proof}
By definition,
 \begin{align}
    \mathbb E_{\mathbf X}\left[\alpha(\mathbf X)\right]=\sum_{k=1}^{K}q_k P_\mathbf{X}(A_k), \quad \beta_r = \sup_{A\in\mathcal{A}_r}P_\mathbf{X}(A).
 \end{align}
Therefore, each $A_k$ satisfies $P_{\mathbf{X}}(A_k)\le \beta_r$. It follows that
\begin{align}
    \mathbb E_{\mathbf X}\left[\alpha(\mathbf X)\right]\le\sum_{k=1}^{K}q_k \beta_r=\beta_r
\end{align}
Combining the convexity of the negative logarithm with Jensen's inequality, for any $\alpha \in \operatorname{VP}_r$, we can derive
\begin{align}
    \mathbb E_{\mathbf X}\left[-\log_2\alpha(\mathbf X)\right]\ge-\log_2 \mathbb E_{\mathbf X}\left[\alpha(\mathbf X)\right]\ge -\log_2 \beta_r.
\end{align}
Taking the infimum over
$\alpha\in\operatorname{VP}_r$ yields the desired result.
\end{proof}
Together, these results provide lower and upper bounds
on $R(r)$. To investigate when these bounds coincide,
we next specialize to homogeneous manifolds.

%---------------%

\section{Special Case}
In this section, we consider a special class of manifolds
to establish the tightness of the upper and lower bounds
and illustrate how they can be evaluated in closed form.
\subsection{Compact Homogeneous Riemannian Manifold}
We first establish a proposition showing that
$\overline{\operatorname{VP}}_r$ contains integral mixtures
of indicators of admissible regions, in addition to
finite convex combinations.
\begin{proposition}
\label{prop_integral_coverage_closure}
Let $\mathcal Z$ be a measurable space equipped with
a probability measure $\nu$. For each $z\in\mathcal Z$,
it marks an element $A_z\in\mathcal A_r$, with the mapping $(z,\mathbf x)\mapsto\mathbf1_{A_z}(\mathbf x)$ being $\mathcal F\otimes\mathcal B(\mathcal M)$-measurable. Define
\begin{align}
\varphi(\mathbf x):=\int_{\mathcal Z}\mathbf 1_{A_z}(\mathbf x)\nu(dz).
\end{align}
Then
\begin{align}
    \varphi \in\overline{\operatorname{VP}}_r.
\end{align}
\end{proposition}

\begin{proof}
To verify the defining approximation property of
$\overline{\operatorname{VP}}_r$, define the sequence
of random functions $\Phi_n$ by
\begin{align}
    \varPhi _n(\mathbf{x}):=\frac{1}{n}\sum_{i=1}^{n}\mathbf 1_{A_{Z_i}}(\mathbf{x}),
\end{align}
where $\mathbf{Z}_1,...,\mathbf{Z}_n$ are independent and identically distributed (i.i.d) random variables with distribution $\nu$.
Note that %Finally, to assess the tightness of the upper bound derived for the general case in Theorem \ref{thm_general_alpha_achievability}, we introduce the results in \cite{Harangi2023GraphonEntropy} that 
$\mathbf 1_{A_{z_i}}(\mathbf{x})$ has a Bernoulli distribution and $\mathbb{E}_{\mathbf{Z}_i}(\mathbf 1_{A_{z_i}}(\mathbf{x}))=\varphi(\mathbf{x})$, for all $i=1,2,...,n$. 
Therefore, 
\begin{align}
\mathbb{E}_{\mathbf{Z}_1,...,\mathbf{Z}_n} \left[\mathbb{E}_\mathbf{X}\left(|| \varPhi_n(\mathbf{X})- \varphi(\mathbf{X})||^2\right)\right]
=\mathbb{E}_\mathbf{X}\left\{\mathbb{E}_{\mathbf{Z}_1,...,\mathbf{Z}_n}\left[\frac{1}{n}\sum_{i=1}^{n}\left(\mathbf 1_{A_{Z_i}}(\mathbf{X})- \varphi(\mathbf{X})\right)\right]^2\right\}
=\mathbb{E}_\mathbf{X}\left(\frac{ \varphi(\mathbf{X})(1- \varphi(\mathbf{X}))}{n}\right)
\le \frac{1}{4n}.
\end{align}
Thus,for each $n$, there exist realizations $\{\mathbf{z}_i\}_{i=1}^{n}$ such that the function $\varphi_n(\mathbf{x})$ satisfies 
\begin{align}
    \varphi_n(\mathbf{x})=\frac{1}{n}\sum_{i=1}^{n}\mathbf 1_{A_{z_i}}(\mathbf{x})
\end{align}
satisfying 
\begin{align}
    \mathbb{E}_\mathbf{X}\left(||\varphi_n(\mathbf{X})- \varphi(\mathbf{X})||^2\right)
    \le\mathbb{E}_{\mathbf{Z}_1,...,\mathbf{Z}_n} \left[\mathbb{E}_\mathbf{X}\left(||\varPhi_n(\mathbf{X})- \varphi(\mathbf{X})||^2\right)\right]
    \le\frac{1}{4n}
\end{align}
Therefore, by the Cauchy-Schwarz inequality, we have 
\begin{align}
    \lim_{n\rightarrow\infty}\mathbb{E}_\mathbf{X}\left(|\varphi_n(\mathbf{X})- \varphi(\mathbf{X})|\right)\le \lim_{n\rightarrow\infty}\frac{1}{2\sqrt{n}}=0,
\end{align}
By construction, $\varphi_n\in\operatorname{VP}_r$.
Since $\varphi_n$ converges to $\varphi$ in
$L^1(P_{\mathbf X})$, the definition of the closure gives
$\varphi\in\overline{\operatorname{VP}}_r$.
% Then we can calculate it value by 
% \begin{align}
%          \varphi(\mathbf{x})&=\int_{\mathbb{S}^d} \mathbf{1}_{B_{r/2}(\mathbf{z})}(\mathbf{x})P_\mathbf{X}(d\mathbf{z})=P_\mathbf{X}(B_{r/2}(\mathbf{x}))\\
%         &=\frac{V_g(B_{r/2}(\mathbf{x}))}{V_g(\mathcal{M})}=\frac{\displaystyle\int_0^{\frac{r}{2}}\sin^{d-1}tdt}{\displaystyle\int_0^\pi\sin^{d-1}tdt}=\beta_r.
% \end{align}
\end{proof}
The proposition shows that $\overline{\operatorname{VP}}_r$ contains integral mixtures
of indicators of admissible regions as well as finite
convex combinations. This result allows us to derive
an explicit upper bound for homogeneous manifolds.

\begin{theorem}\label{thm:homogeneous_upper_bound}
Let $(\mathcal M,g)$ be a compact homogeneous Riemannian manifold, and let $P_{\mathbf X}$ be the uniform distribution on $\mathcal M$. For any $r>0$, 
\begin{align}
R(r)=-\log_2\beta_r.
\end{align}
\end{theorem}
\begin{proof}
Since $(\mathcal M,g)$ is compact and homogeneous,
there exists a compact group of isometries $G$ acting
transitively on $\mathcal M$. Let $\mu_G$ denote the
normalized Haar probability measure on $G$, so that
$\mu_G(G)=1$.
Then, fix $\eta\in(0,\beta_r)$. By the definition of $\beta_r$, there exists a nearly optimal Borel set $A_\eta\subseteq\mathcal M$ such that
\begin{align}
\operatorname{diam}_g(A_\eta)\le r,\qquad p_\eta:= P_{\mathbf X}(A_\eta)>\beta_r-\eta.
\end{align}

To apply Proposition~\ref{prop_integral_coverage_closure},
take $\mathcal Z=G$, $\nu=\mu_G$, and
$A_u=uA_\eta$ for $u\in G$, where
$uA_\eta:=\{u\mathbf x:\mathbf x\in A_\eta\}$.
Since each $u\in G$ is an isometry,
\begin{align}
    \operatorname{diam}_g(\mu A_\eta)=\operatorname{diam}_g(A_\eta)\le r
\end{align}
Therefore, $\mu A_\eta \in \mathcal{A}_r$.
Hence, we can define
\begin{align}
\alpha_\eta^*(\mathbf{x}):=\int_{G} \mathbf{1}_{\mu A_\eta}(\mathbf{x})\mu_G(d\mu).
\end{align}
Then, for any $\eta>0$, we have
\begin{align}
    \alpha_\eta^* \in\overline{\operatorname{VP}}_r
\end{align}
As $\mathcal M$ is homogeneous, for arbitrary $\mathbf x,\mathbf \mathbf{x}'\in\mathcal M$, there exists $\sigma \in G$ such that $\mathbf{x}'=\sigma \mathbf x$. Hence,
\begin{align}
\alpha_\eta^*(\mathbf{x}')&=\int_G\mathbf 1\left\{\sigma\mathbf x\in uA_\eta\right\}\mu_G(du)
=\int_G\mathbf 1\left\{\mathbf x\in \sigma^{-1}uA_\eta\right\}\mu_G(du)
=\int_G\mathbf 1\left\{\mathbf x\in wA_\eta\right\}\mu_G(dw)
=\alpha_\eta^*(\mathbf{x}).
\end{align}
The change of variables $w=\sigma^{-1}u$ and the left
invariance of Haar measure give the third equality. Thus, $\alpha_\eta^*(\mathbf{x})$ is constant on $\mathcal M$. Since $P_{\mathbf X}$ is a probability measure,
\begin{align}
\alpha_\eta^*(\mathbf{x})&=\mathbb{E}_\mathbf{X}\left[\alpha_\eta^*(\mathbf{X})\right]
=\int_{\mathcal M}\alpha_\eta^*(\mathbf{x})P_{\mathbf X}(d\mathbf{x})
=\int_G\int_{\mathcal M}\mathbf 1\left\{\mathbf{x}\in uA_\eta\right\}P_{\mathbf X}(d\mathbf{x})\mu_G(du)
=\int_G P_{\mathbf X}(uA_\eta)\mu_G(du)
\end{align}

Since $P_{\mathbf X}$ is uniform and isometries preserve
Riemannian volume, for every $u\in G$,
\begin{align}
    P_\mathbf{X}(\mu A_\eta)=P_\mathbf{X}(A_\eta)=p_\eta.
\end{align}
Thus, we have
\begin{align}
    \alpha_\eta^*(\mathbf{x})=p_\eta > \beta_r-\eta.
\end{align}
Substituting this coverage function into the upper bound
in Theorem~\ref{thm_general_alpha_achievability} gives
\begin{align}
    R(r)\le \mathbb E_{{\mathbf{X}}}\left[-\log_2\alpha_\eta^*({\mathbf{X}})\right]<-\log_2(\beta_r -\eta).
\end{align}

Letting $\eta\downarrow0$ yields
\begin{align}
    R(r)\le -\log_2 \beta_r.
\end{align}
Combining this inequality with
Theorem~\ref{thm_generative_rate_distortion_converse}
completes the proof.
\end{proof}
The proof establishes that the lower and upper bounds
coincide and provides a construction whose mutual
information approaches the optimal rate.

\begin{corollary}
For a uniformly distributed source on a compact
homogeneous Riemannian manifold, the optimal rate
$R(r)=-\log_2\beta_r$ can be approached arbitrarily closely
using the encoder and generator kernels
    \begin{align}
        &P_{\hat U\mid\mathbf X}(du\mid\mathbf x)=\frac{\mathbf 1\!\left\{\mathbf x\in uA_\eta\right\}}{p_\eta}\mu_G(du),\\
        &P_{\hat{\mathbf X}\mid\hat U}(d\hat{\mathbf x}\mid u):=\frac{\mathbf 1\!\left\{\hat{\mathbf x}\in uA_\eta\right\}}{p_\eta}P_{\mathbf X}(d\hat{\mathbf x}).
    \end{align}
\end{corollary}
\begin{proof}
By construction, the joint distribution of $(\mathbf X,\hat U)$ is
\begin{align}
P_{\mathbf X,\hat U}(d\mathbf x,du)=\frac{\mathbf 1\!\left\{\mathbf x\in uA_\eta\right\}}{p_\eta}P_{\mathbf X}(d\mathbf x)\mu_G(du).
\end{align}
Consequently,
\begin{align}
P_{\mathbf X\mid\hat U}(d\mathbf x\mid u)=\frac{\mathbf 1\!\left\{\mathbf x\in uA_\eta\right\}}{p_\eta}P_{\mathbf X}(d\mathbf x).
\end{align}
Therefore, we have
\begin{align}
P_{\hat{\mathbf X}\mid\hat U=u}=P_{\mathbf X\mid\hat U=u},\qquad P_{\hat U}\text{-a.e. }u,
\end{align}
and hence the conditional-distribution constraint is satisfied. By construction, $\hat{\mathbf X}$ is generated solely according to $P_{\hat{\mathbf X}\mid\hat U}$, and therefore $\mathbf X-\hat U-\hat{\mathbf X}$ holds.

For $P_{\hat U}$-almost every $u$, we have
\begin{align}
P_{\mathbf X\mid\hat U=u}(uA_\eta)=P_{\hat{\mathbf X}\mid\hat U=u}(uA_\eta)=\int_{uA_\eta}P_{\mathbf X\mid\hat U=u}(dx)=1.
\end{align}
Since
\begin{align}
\operatorname{diam}_g(uA_\eta)\le r,
\end{align}
we have
\begin{align}
P\!\left(d_g(\mathbf X,\hat{\mathbf X})>r\mid\hat U=u\right)=0.
\end{align}
Since $P_{\mathbf X}$ is the uniform distribution with
respect to Riemannian volume, $P_{\mathbf X}\ll V_g$. Hence, both $P_{\mathbf X\mid\hat U=u}$ and $P_{\hat{\mathbf X}\mid\hat U=u}$ are absolutely continuous with respect to $V_g$. For $r>0$, the level set
\begin{align}
\left\{(\mathbf \mathbf{x},\hat{\mathbf x})\in\mathcal M\times\mathcal M:d_g(\mathbf \mathbf{x},\hat{\mathbf x})=r\right\}
\end{align}
has zero $(V_g\otimes V_g)$-measure. Hence,
\begin{align}
P\!\left(d_g(\mathbf X,\hat{\mathbf X})=r\mid\hat U=u\right)=0,
\end{align}
and consequently,
\begin{align}
P\!\left(d_g(\mathbf X,\hat{\mathbf X})\ge r\right)=0.
\end{align}
Finally, we evaluate the required rate. For $P_{\hat U}$-almost every $u$,
\begin{align}
\frac{dP_{\mathbf X\mid\hat U=u}}{dP_{\mathbf X}}(\mathbf x)=\frac{\mathbf 1\!\left\{\mathbf x\in uA_\eta\right\}}{p_\eta}.
\end{align}
Therefore,
\begin{align}
D_{\mathrm{KL}}\!\left(P_{\mathbf X\mid\hat U=u}\middle\|P_{\mathbf X}\right)=-\log_2 p_\eta.
\end{align}
Using the conditional-divergence representation of mutual information,
\begin{align}
I(\mathbf X;\hat U)=\int_G D_{\mathrm{KL}}\!\left(P_{\mathbf X\mid\hat U=u}\middle\|P_{\mathbf X}\right)P_{\hat U}(du)=\int_G\left(-\log_2 p_\eta\right)\mu_G(du)=-\log_2 p_\eta.
\end{align}
Since $\eta>0$ is arbitrary, letting $\eta\downarrow0$ yields
\begin{align}
I(\mathbf X;\hat U)=-\log_2\beta_r.
\end{align}
\end{proof}

\subsection{$d$-Dimensional Spherical Manifold}
We now specialize to a uniformly distributed source
on a spherical manifold. This example also illustrates
how a manifold with low intrinsic dimension can be
embedded in a high-dimensional ambient space.

First, we define the low-dimensional manifold $\mathcal{M}$. Let $\mathbb{S}^d$ denote the standard $d$-dimensional unit sphere, defined by
\begin{align}
    \mathbb S^d:=\left\{\mathbf z\in\mathbb R^{d+1}:||\mathbf z||_2=1\right\},
\end{align}
equipped with the standard round Riemannian metric $g^{\mathbb{S}}$, inducing the geodesic distance $d_g^{\mathbb{S}}$ and volume measure $V_{g}^{\mathbb{S}}$.
Further, let $\mathbf{L} \in \mathbb{R}^{m\times(d+1)}$ be a semi-orthogonal matrix satisfying
\begin{align}
\mathbf{L}^{\mathsf T}\mathbf{L}=\mathbf I_{d+1},
\end{align}
where $\mathbf I_{d+1}$ is the $(d+1)$-dimensional identity matrix.
Now, the low-dimensional manifold $\mathcal M$ embedded in the high-dimensional ambient space $\mathbb{R}^m$ can be defined by
\begin{align}
\mathcal M:=\left\{\mathbf x\in\mathbb{R}^m:\mathbf{x}=\mathbf{L}\mathbf{z},\mathbf{z}\in\mathbb{R}^{d+1}, ||\mathbf z||_2=1\right\},
\end{align}
equipped with the Riemannian metric $g$ induced by $g^{\mathbb{S}}$ and $\mathbf{L}$. Specifically, for any $\mathbf{x}=\mathbf{L}\mathbf{z}\in\mathcal M$ and any $\mathbf{v},\mathbf{w}\in T_\mathbf{x}\mathcal M$, the Riemannian metric $g$ is defined by
\begin{align}
    g_\mathbf{x}(\mathbf{v},\mathbf{w})=g^{\mathbb{S}}_\mathbf{z}(\mathbf{L}^{\mathsf T}\mathbf{v},\mathbf{L}^{\mathsf T}\mathbf{w})=\mathbf{v}^{\mathsf T}\mathbf{L}\mathbf{L}^{\mathsf T}\mathbf{w}=\mathbf{v}^{\mathsf T}\mathbf{w}.
\end{align}
Therefore, for $\mathbf x=\mathbf{L}\mathbf{z}, \hat{\mathbf{x}}=\mathbf{L}\hat{\mathbf{z}}\in\mathcal M$, the geodesic distance is preserved as
\begin{align}
    d_g(\mathbf x, \hat{\mathbf{x}})=d_g^{\mathbb{S}}(\mathbf z, \hat{\mathbf{z}})=\arccos(\mathbf z^{\mathsf T} \hat{\mathbf{z}}).
\end{align}
Similarly, the volume measure is also preserved.
Thus, although each point of $\mathcal M$ is represented
by $m$ ambient coordinates, the manifold has only $d$
intrinsic degrees of freedom.

We assume that $\mathbf X$ is uniformly distributed
on $\mathcal M$ with respect to $V_g$.
In this case, evaluating $\beta_r$ reduces to finding
the maximum volume of a set satisfying the diameter
constraint:
\begin{align}
    \beta_r=\sup_{{A\in \mathcal{B}(\mathcal{M}):\operatorname{diam}_g(A)\le r}}\frac{V_g(A)}{V_g(\mathcal{M})}.
\end{align}
As the manifold $(\mathcal{M},g)$ is induced by $\mathbb{S}$ and a semi-orthogonal matrix $\mathbf{L}$, the geodesic distance and volume measure are preserved. Therefore, we have
\begin{align}
    \beta_r=\sup_{{A\in \mathcal{B}(\mathbb{S}^d):\operatorname{diam}_{g^\mathbb{S}}(A)\le r}}\frac{V_g^{\mathbb{S}}(A)}{V_g^{\mathbb{S}}(\mathbb{S}^d)}.
\end{align}
By the spherical isodiametric inequality \cite{BoroczkySagmeister2020}, for $d\ge 2$ and $r\in(0,\pi)$, if a measurable set $A$ satisfies $\operatorname{diam}_{g^\mathbb{S}}(A)\le r$, then
\begin{align}
    V_g^{\mathbb{S}}(A)\le V_g^{\mathbb{S}}(B^\mathbb{S}_{r/2}(\cdot)).
\end{align}
Therefore, the maximum volume of feasible set is $V_g^{\mathbb{S}}(B^\mathbb{S}_{r/2}(\cdot))$.
On the manifold, the volume of a geodesic ball of radius $r$, satisfying $0<r<\pi$, is given by
\begin{align}\label{eq:spherical_volume}
V_g^{\mathbb{S}}(B^\mathbb{S}_r(\cdot))=|\mathbb{S}^{d-1}|\int_{0}^{r}sin^{d-1}\theta d\theta.
\end{align}
Thus, we have
\begin{align}
    \beta_r=\frac{V_g^{\mathbb{S}}(B^\mathbb{S}_{r/2}(\cdot))}{V_g^{\mathbb{S}}(\mathbb{S}^d)}=\frac{\displaystyle\int_0^{\frac{r}{2}}\sin^{d-1}tdt}{\displaystyle\int_0^\pi\sin^{d-1}tdt}.
\end{align}
Substituting this expression into
Theorem~\ref{thm:homogeneous_upper_bound} yields
\begin{align}\label{eq:eq:spherical_rate_lower_bound}
R(r)=-\log_2\left[\frac{\displaystyle\int_0^{\frac{r}{2}}\sin^{d-1}t dt}{\displaystyle\int_0^\pi\sin^{d-1}t dt}\right].
\end{align}
Fig.\ref{fig_small_radius_approximation} shows the rate $R(r)$ versus the radius $r$ for different dimensions $d$. It can be observed that the rate increases as the radius decreases, and the rate also increases with the dimension of the manifold.
%This result further corroborates Corollary \ref{Cor_general_rate_distortion_bounds} and, at the same time, confirms the consistency of the upper bound established in Theorem \ref{thm:homogeneous_upper_bound} and Theorem \ref{thm_general_alpha_achievability}.

\begin{remark}
In this case, when $d\ge 2$ and $0<r<\pi$, $R(r)$ is strictly decreasing in $r$. We have
\begin{align}
0<R\left(\frac r2\right)-R(r)
=\log_2\frac{V_g(B_{r/2})}{V_g(B_{r/4})}<d.
\end{align}
Therefore, increasing the communication rate by $d$ bits per source symbol can reduce the distortion radius by more than half. Furthermore, we have
\begin{align}
\lim_{r\downarrow0}\left[R\left(\frac r2\right)-R(r)\right]=d.
\end{align}
Thus, the rate--radius relationship asymptotically approaches the scaling law that an additional $d$ bits halve the distortion radius.
\end{remark}
\begin{figure}[t]
\centering
\includegraphics[width=0.48\textwidth]{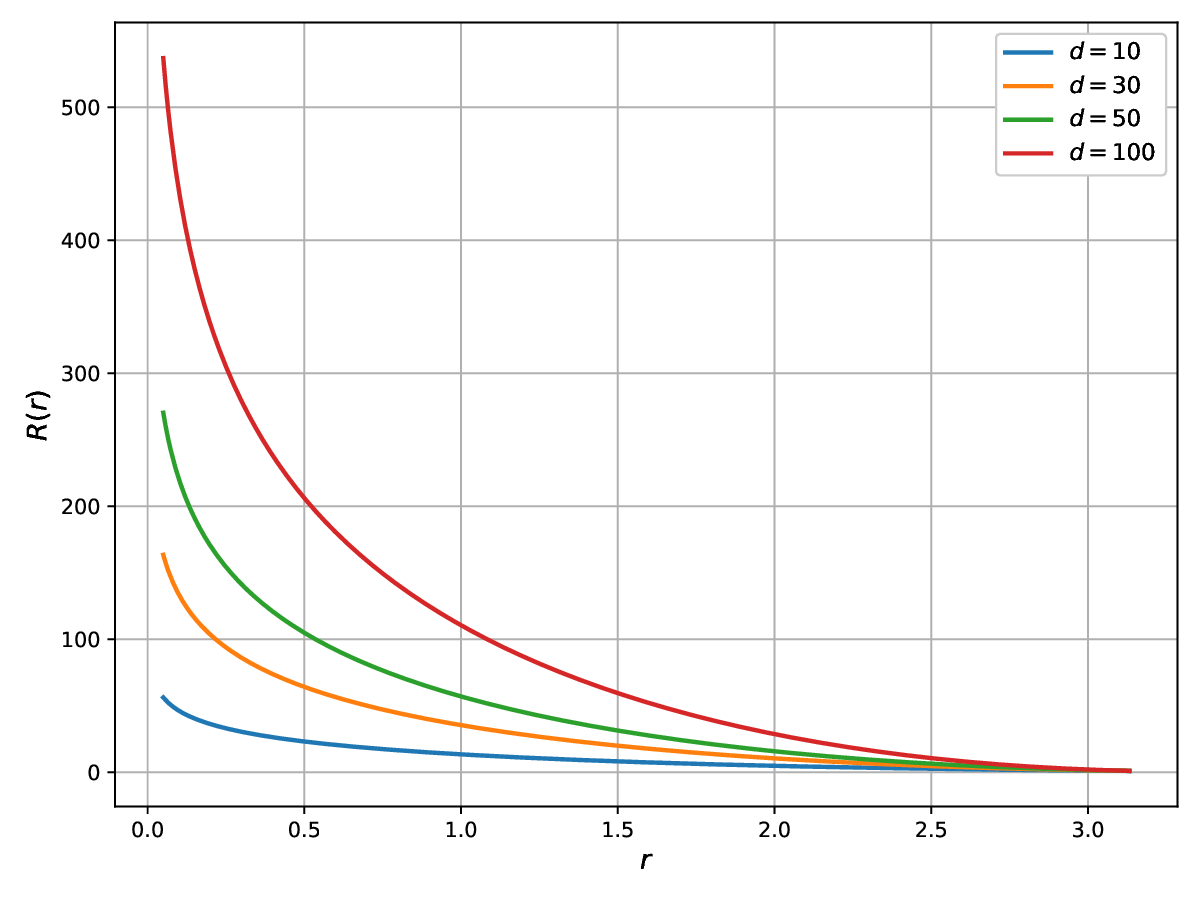}
\caption{The rate $R(r)$ versus radius $r$ for the spherical manifold with different intrinsic dimensions.}
\label{fig_small_radius_approximation}
\end{figure}

\section{Conclusion}
This paper investigates how much source-specific
information must be communicated to support reliable
stochastic generation while preserving the conditional
source distribution. For sources on compact Riemannian
manifolds, we formulate an operational GRDP function
under geodesic distortion and failure-probability
constraints without common randomness, and establish
its single-letter characterization. This characterization
quantifies the minimum communication rate required
to constrain generation to meet prescribed fidelity
and reliability requirements.

In the high-reliability limit, we derive a geometric
lower bound determined by the largest source probability
mass of a region with bounded geodesic diameter, together
with an upper bound based on probabilistic coverings
of the source manifold. Small-radius analysis of the
lower bound identifies intrinsic dimension as governing
its leading rate scaling and scalar curvature as
contributing a correction. For uniformly distributed
sources on compact homogeneous Riemannian manifolds,
the bounds coincide, yielding an exact characterization
that can be explicitly evaluated for spherical manifolds.
These results connect the information required for
reliable generation to the geometric concentration
of the source distribution, clarifying how communication
constrains the uncertainty of stochastic generation.

\appendices
\section{Proof of Theorem~\ref{theorem:coding_theorem}}\label{app:coding_theorem}
\begin{lemma}[Finite-Alphabet Posterior-Reference Coding]\label{lem:posterior-reference}
Let $P_V$ and $P_{Y|V}$ denote the output distribution and the backward test channel induced by $P_Y$ and $P_{V|Y}$, respectively, where $\mathcal Y$ and $\mathcal V$ are finite alphabets. For any $\delta>0$, for all sufficiently large $n$, there exists a stochastic mapping $\psi^{(n)}:\mathcal Y^n\to G^{(n)}\cup B^{(n)}\subseteq\mathcal V^n$
mapping $Y^n$ to $V^n$ according to some conditional distribution $P_{V^n|Y^n}$, with 
$G^{(n)}\subseteq T_{\delta}^{(n)}(P_V)$ and $G^{(n)}\cap B^{(n)}=\varnothing,$
such that
\begin{align}
&\frac{1}{n}\left\lceil\log_2\left|G^{(n)}\cup B^{(n)}\right|\right\rceil\le I(Y;V)+\delta,\\
&\Pr\{V^n\in B^{(n)}\}\le\delta,\\
&d_{\mathrm{TV}}\left(P_{Y^n|V^n}(\cdot|v^n),P_{Y|V}^{\otimes n}(\cdot|v^n)\right)\le\delta, \quad v^n \in G^n
\end{align}
Here,
\begin{align}
P_{Y|V}^{\otimes n}(y^n|v^n):=\prod_{t=1}^{n}P_{Y|V}(y_t|v_t).
\end{align}
\end{lemma}
\begin{proof}
    The lemma is proven in \cite{salehkalaibar2024rdp,atif2024posterior}
\end{proof}

We first prove
\begin{align}
R_{\mathrm C}(r,\epsilon)\le R(r,\epsilon).
\end{align}

%Fix an arbitrary $\gamma>0$. We shall construct a sequence of stochastic block codes satisfying the operational constraints in Definition \ref{def:achieveable_rate} and having asymptotic rate no larger than $R(r,\epsilon)+\gamma$.
\textbf{Step A: Introduce a strict quantile slack.}

%We first establish that a strict slack in the quantile constraint can be introduced without asymptotic rate loss. 
Consider two feasible single-letter tuples $({\mathbf{X}},\hat U_1,\hat {\mathbf{X}}_1)$ and $({\mathbf{X}},\hat U_2,\hat {\mathbf{X}}_2)$ satisfying
$\Pr\{d_g({\mathbf{X}},\hat {\mathbf{X}}_i)\ge r\}\le\epsilon_i,\ i\in\{1,2\}.$
Let $Q\in\{1,2\}$ be independent of ${\mathbf{X}}$ with
$\Pr\{Q=1\}=1-\lambda$ and $\Pr\{Q=2\}=\lambda.$
Define $\hat U':=(Q,\hat U_Q)$ and $\hat {\mathbf{X}}':=\hat {\mathbf{X}}_Q.$
Then
\begin{align}
\Pr\{d_g({\mathbf{X}},\hat {\mathbf{X}}')\ge r\}\le(1-\lambda)\epsilon_1+\lambda\epsilon_2.
\end{align}
Moreover, since $Q$ is independent of ${\mathbf{X}}$,
\begin{align}
I({\mathbf{X}};\hat U')=(1-\lambda)I({\mathbf{X}};\hat U_1)+\lambda I({\mathbf{X}};\hat U_2).
\end{align}
The conditional-distribution preservation constraint is also satisfied. Hence, $R(r,\epsilon)$ is convex in $\epsilon$.
 
Furthermore, since $(\mathcal M,g)$ is compact and $d_g$ is the geodesic distance induced by $g$, $(\mathcal M,d_g)$ is a compact metric space and hence totally bounded. Therefore, for every $r>0$, there exists a finite Borel partition $\{C_1,\ldots,C_N\}$ such that
\begin{align}
\max_{1\le i\le N}\operatorname{diam}_{d_g}(C_i)<r.
\end{align}
Let $\hat U$ be the corresponding cell index and define the reconstruction according to
$P_{\hat {\mathbf{X}}|\hat U}=P_{{\mathbf{X}}|\hat U}.$
Conditioned on $\hat U$, both ${\mathbf{X}}$ and $\hat {\mathbf{X}}$ lie in the same partition cell almost surely. Hence,
\begin{align}
\Pr\{d_g({\mathbf{X}},\hat {\mathbf{X}})\ge r\}=0.
\end{align}
Furthermore,
\begin{align}
I({\mathbf{X}};\hat U)\le H(\hat U)\le\log_2N<\infty.
\end{align}
Consequently,
\begin{align}
R(r,0)<\infty.
\end{align}
Since $R(r,\epsilon)$ is finite and convex on $(0,1)$, it is continuous on $(0,1)$. In particular,
\begin{align}
\lim_{\xi\downarrow0}R(r,\epsilon-\xi)=R(r,\epsilon).
\end{align}

Therefore, there exists $\epsilon_0\in(0,\epsilon)$ such that
\begin{align}
R(r,\epsilon_0)\le R(r,\epsilon)+\frac{ \varsigma  }{4}.
\end{align}
By the definition of the infimum, there exists a feasible single-letter tuple $({\mathbf{X}},\hat U,\hat {\mathbf{X}})$ satisfying
\begin{align}
&{\mathbf{X}}-\hat U-\hat {\mathbf{X}},\\
&P_{\hat {\mathbf{X}}|\hat U}=P_{{\mathbf{X}}|\hat U}\qquad P_{\hat U}\text{-a.s.},\\
&\Pr\{d_g({\mathbf{X}},\hat {\mathbf{X}})\ge r\}\le\epsilon_0,\\
&I({\mathbf{X}};\hat U)\le R(r,\epsilon_0)+\frac{ \varsigma}{4}\le R(r,\epsilon)+\frac{ \varsigma}{2}.
\end{align}
We fix this tuple for the remainder of the achievability proof.

\textbf{Step B: quantize the compact manifold}. 

Since $(\mathcal M,d_g)$ is compact, there exists a sequence of finite Borel partitions
\begin{align}
\mathcal P_k=\{C_{k,1},\ldots,C_{k,N_k}\},
\end{align}
whose maximal geodesic diameters
\begin{align}
a_k:=\max_{1\le i\le N_k}\operatorname{diam}_{d_g}(C_{k,i})
\end{align}
satisfy
\begin{align}
\lim_{k\to\infty}a_k=0.
\end{align}
Let
\begin{align}
Q_k:\mathcal M\to\mathcal Y_k:=\{1,\ldots,N_k\}
\end{align}
denote the corresponding quantizer, and define
\begin{align}
Y_k:= Q_k({\mathbf{X}}),\qquad \hat Y_k:= Q_k(\hat {\mathbf{X}}).
\end{align}

Since $P_{\hat {\mathbf{X}}|\hat U}=P_{{\mathbf{X}}|\hat U}$, for every $i\in\mathcal Y_k$,
\begin{align}
P_{\hat Y_k|\hat U=u}(i)=P_{Y_k|\hat U=u}(i)\qquad P_{\hat U}\text{-a.s.}
\end{align}
Hence,
\begin{align}
P_{\hat Y_k|\hat U}=P_{Y_k|\hat U}.
\end{align}
Furthermore, since $Y_k$ is a deterministic function of ${\mathbf{X}}$, the data-processing inequality gives
\begin{align}
I(Y_k;\hat U)\le I({\mathbf{X}};\hat U),
\end{align}
and consequently,
\begin{align}
Y_k-\hat U-\hat Y_k,
\end{align}
\textbf{Step C: Define the cost on the quantized space $\mathcal{Y}_k$}

We define the cost on $\mathcal{Y}_k$ that upper-bounds the original excess-distortion indicator.

For each $k$, define 
\begin{align}
\Delta_{r,k}(i,j):=\mathbf{1}\left\{\sup_{\mathbf{x}\in C_{k,i}, \hat{\mathbf{x}}\in C_{k,j}}d_g(\mathbf{x},\hat{\mathbf{x}})\ge r\right\}.
\end{align}
Therefore, for every $\mathbf{x},\hat{\mathbf{x}}\in\mathcal M$,
\begin{align}
\mathbf{1}\{d_g(\mathbf{x},\hat{\mathbf{x}})\ge r\}\le\Delta_{r,k}(Q_k(\mathbf{x}),Q_k(\hat{\mathbf{x}})).
\end{align}

We next establish the pointwise convergence of the robust threshold loss. Fix $(\mathbf{x},\hat{\mathbf{x}})\in\mathcal M\times\mathcal M$. If $d_g(\mathbf{x},\hat{\mathbf{x}})\ge r$, then
\begin{align}
\mathbf{1}\{d_g(\mathbf{x},\hat{\mathbf{x}})\ge r\}=1, \ \Delta_{r,k}(Q_k(\mathbf{x}),Q_k(\hat{\mathbf{x}}))=1,\ \forall k.
\end{align}
If, on the other hand, $d_g(\mathbf{x},\hat{\mathbf{x}})<r$, meaning that $\mathbf{1}\{d_g(\mathbf{x},\hat{\mathbf{x}})\ge r\}=0$, we define
\begin{align}
\eta(\mathbf{x},\hat{\mathbf{x}}):= r-d_g(\mathbf{x},\hat{\mathbf{x}})>0.
\end{align}
Since $\lim_{k\to\infty}a_k=0$, there exists $K(\mathbf{x},\hat{\mathbf{x}})$ such that, for every $k\ge K(\mathbf{x},\hat{\mathbf{x}})$,
\begin{align}
2a_k<\eta(\mathbf{x},\hat{\mathbf{x}}).
\end{align}
For every $\mathbf{x}'\in C_{k,Q_k(\mathbf{x})}$ and $\hat{\mathbf{x}}'\in C_{k,Q_k(\hat{\mathbf{x}})}$, the triangle inequality for the geodesic distance gives
\begin{align}
d_g(\mathbf{x}',\hat{\mathbf{x}}')&\le d_g(\mathbf{x}',\mathbf{x})+d_g(\mathbf{x},\hat{\mathbf{x}})+d_g(\hat{\mathbf{x}},\hat{\mathbf{x}}')\leq d_g(\mathbf{x},\hat{\mathbf{x}})+2a_k<r.
\end{align}
Hence,
\begin{align}
\Delta_{r,k}(Q_k(\mathbf{x}),Q_k(\hat{\mathbf{x}}))=0
\end{align}
for all sufficiently large $k$. Therefore,
\begin{align}
\lim_{k\to\infty}\Delta_{r,k}(Q_k(\mathbf{x}),Q_k(\hat{\mathbf{x}}))=\mathbf{1}\{d_g(\mathbf{x},\hat{\mathbf{x}})\ge r\}.
\end{align}
Since $0\le\Delta_{r,k}\le1$, the bounded convergence theorem yields
\begin{align}
\lim_{k\to\infty}\mathbb E[\Delta_{r,k}(Y_k,\hat Y_k)]=\Pr\{d_g({\mathbf{X}},\hat {\mathbf{X}})\ge r\}.
\end{align}
Since the right-hand side is at most $\epsilon_0<\epsilon$, there exists a finite $k$ such that
\begin{align}
D_k:=\mathbb E[\Delta_{r,k}(Y_k,\hat Y_k)]<\epsilon.
\end{align}
Fix such a $k$, and write
\begin{align}
Y:= Y_k,\qquad \hat Y:=\hat Y_k,\qquad \Delta:=\Delta_{r,k}.
\end{align}
Then
\begin{align}
P_{\hat Y|\hat U}=P_{Y|\hat U},
\end{align}
\begin{align}
\mathbb E[\Delta(Y,\hat Y)]=D_k<\epsilon,
\end{align}
and
\begin{align}
I(Y;\hat U)\le I({\mathbf{X}};\hat U).
\end{align}

\textbf{Step D: replace the possibly nonfinite auxiliary variable $\hat U$ by a finite auxiliary variable}

For each $u$, define
\begin{align}
q_u(y):= P_{Y|\hat U=u}(y), \ h(q):=-\sum_{y\in\mathcal Y}q(y)\log_2q(y).
\end{align}
Since $P_{\hat Y|\hat U}=P_{Y|\hat U}$ and $Y-\hat U-\hat Y$, conditioned on $\hat U=u$, the random variables $Y$ and $\hat Y$ are independent and identically distributed according to $q_u$. Define also
\begin{align}
c(q_u):=\mathbb E[\Delta(Y,\hat Y)\mid \hat U=u]=\sum_{y,\hat y\in\mathcal Y}q_u(y)q_u(\hat y)\Delta(y,\hat y).
\end{align}
Then
\begin{align}
H(Y|\hat U)=\mathbb E[h(q_{\hat U})], \qquad D_k=\mathbb E[c(q_{\hat U})].
\end{align}

We only need to perserve the $N_k+1$ continuous functions consisting of the first $N_k-1$ coordinates of $q$, together with $h(q)$ and $c(q)$. By the support lemma \cite[Appendix~C]{ElGamalKim2011}, there exists a finite auxiliary random variable $V$ satisfying
\begin{align}
|\mathcal V|\le N_k+2
\end{align}
such that the marginal distribution $P_Y$, the conditional entropy, and the expected cost are preserved. In particular,
\begin{align}
H(Y|V)=H(Y|\hat U).
\end{align}
Hence,
\begin{align}
I(Y;V)=I(Y;\hat U).
\end{align}
Define the reconstruction conditionally on $V$ according to
\begin{align}
P_{\hat Y|V}=P_{Y|V},
\end{align}
with $Y$ and $\hat Y$ conditionally independent given $V$. By construction,
\begin{align}
\mathbb E[\Delta(Y,\hat Y)]=D_k.
\end{align}
Consequently,
\begin{align}
I(Y;V)=I(Y;\hat U)\le I({\mathbf{X}};\hat U).
\end{align}

%------------------%
\textbf{Step E: achievability on the quantized space}
By Lemma~\ref{lem:posterior-reference}, we can define $\mathcal{C}^{(n)}:=G^{(n)}\cup B^{(n)}$, and assign a distinct binary message $E\in\mathcal E$ to each codeword $v^n\in\mathcal C^{(n)}$. Then
\begin{align}
\frac{1}{n}\mathbb E[\ell(E)]\le I(Y;V)+\delta.
\end{align}

Given $V^n=v^n$, define the decoder by posterior resampling,
\begin{align}
P_{\hat Y^n|V^n=v^n}:=P_{Y^n|V^n=v^n}.
\end{align}
Since $E$ and $V^n$ are in one-to-one correspondence,
\begin{align}
P_{\hat Y^n|E}=P_{Y^n|E}\qquad P_{E}\text{-a.s.}
\end{align}

It remains to verify the finite-alphabet quantile constraint. Define
\begin{align}
F_n(Y^n,\hat Y^n):=\frac{1}{n}\sum_{t=1}^{n}\Delta(Y_t,\hat Y_t).
\end{align}
such that 
\begin{align}
    &\mathbb E_{P_{Y^n\hat Y^n|V^n=v^n}}\left[F_n\right]=\mathbb E[F_n|V^n=v^n]\\
    &\mathbb E_{P_{Y|V}^{\otimes n}(\cdot|v^n)\otimes P_{\hat{Y}|V}^{\otimes n}(\cdot|v^n)}\left[F_n\right]=\frac{1}{n}\sum_{t=1}^{n}\mathbb E[\Delta(Y,\hat Y)|V=v_t].\nonumber
\end{align}
For $v^n\in G^{(n)}$, posterior resampling and the total-variation bound in Lemma~\ref{lem:posterior-reference} imply
\begin{align}
d_{\mathrm{TV}}\left(P_{Y^n\hat Y^n|V^n=v^n},P_{Y|V}^{\otimes n}(\cdot|v^n)\otimes P_{\hat{Y}|V}^{\otimes n}(\cdot|v^n)\right)\le2\delta.
\end{align}
Therefore, since $0\le F_n\le1$,
\begin{align}
\mathbb E[F_n|V^n=v^n]\le\frac{1}{n}\sum_{t=1}^{n}\mathbb E[\Delta(Y,\hat Y)|V=v_t]+2\delta.
\end{align}
Since $v^n\in G^{(n)}\subseteq T_{\delta}^{(n)}(P_V)$, the typical average lemma implies that there exists a nonnegative function $\beta(\delta)$ satisfying
\begin{align}
\lim_{\delta\downarrow0}\beta(\delta)=0
\end{align}
such that
\begin{align}
\frac{1}{n}\sum_{t=1}^{n}\mathbb E[\Delta(Y,\hat Y)|V=v_t]\le D_k+\beta(\delta).
\end{align}
Hence, for every $v^n\in G^{(n)}$,
\begin{align}
\mathbb E[F_n|V^n=v^n]\le D_k+\beta(\delta)+2\delta.
\end{align}
For $v^n\in B^{(n)}$, we use only $F_n\le1$. Together with $\Pr\{V^n\in B^{(n)}\}\le\delta$, this gives
\begin{align}
\mathbb E[F_n]\le D_k+\beta(\delta)+3\delta.
\end{align}

Since $D_k<\epsilon$ and
$\lim_{\delta\downarrow0}\left[\beta(\delta)+3\delta\right]=0,$
we may choose $\delta>0$ sufficiently small such that
\begin{align}
D_k+\beta(\delta)+3\delta<\epsilon, \quad \delta<\frac{ \varsigma}{2}.
\end{align}
Therefore, for all sufficiently large $n$,
\begin{align}
\frac{1}{n}\sum_{t=1}^{n}\mathbb E[\Delta(Y_t,\hat Y_t)]<\epsilon.
\end{align}

\textbf{Step F: lift the finite-alphabet reconstruction back to the manifold $\mathcal M$}

Define the probability kernel
\begin{align}
K(B|y):= P_{{\mathbf{X}}|Y}(B|y)
\end{align}
for every Borel set $B\subseteq\mathcal M$ and every $y$ satisfying $P_Y(y)>0$. For $y$ satisfying $P_Y(y)=0$, choose any probability measure supported on $C_{k,y}$. Then
\begin{align}
K(C_{k,y}|y)=1,\qquad \forall y\in\mathcal Y.
\end{align}
Since the source is memoryless and the quantizer acts symbolwise,
\begin{align}
P_{{\mathbf{X}}^n|Y^n=y^n}=\prod_{t=1}^{n}K(dx_t|y_t)
\end{align}
for every $y^n$ having positive probability. Moreover, the finite-alphabet encoder depends on ${\mathbf{X}}^n$ only through $Y^n$, and hence
\begin{align}
{\mathbf{X}}^n-Y^n-E.
\end{align}

After producing $\hat Y^n$, let the decoder independently generate the continuous reconstruction from the same kernel according to
\begin{align}
P_{\hat {\mathbf{X}}^n|\hat Y^n=\hat y^n}:=\prod_{t=1}^{n}K(d\hat{\mathbf{x}}_t|\hat y_t).
\end{align}
For every $e\in\mathcal E$ satisfying $P_{E}(e)>0$,
\begin{align}
P_{{\mathbf{X}}^n|E=e}=\sum_{y^n}P_{Y^n|E=e}(y^n)\prod_{t=1}^{n}K(dx_t|y_t).
\end{align}
Similarly,
\begin{align}
P_{\hat {\mathbf{X}}^n|E=e}=\sum_{\hat y^n}P_{\hat Y^n|E=e}(\hat y^n)\prod_{t=1}^{n}K(d\hat{\mathbf{x}}_t|\hat y_t).
\end{align}
Since $P_{\hat Y^n|E}=P_{Y^n|E}$, we have 
\begin{align}
P_{\hat {\mathbf{X}}^n|E}=P_{{\mathbf{X}}^n|E}\qquad P_{E}\text{-a.s.}
\end{align}
Consequently,
\begin{align}
P_{\hat {\mathbf{X}}^n}=P_{{\mathbf{X}}^n}=P_{\mathbf{X}}^{\otimes n}.
\end{align}

By the construction of the lifting kernel,
\begin{align}
{\mathbf{X}}_t\in C_{k,Y_t},\qquad \hat {\mathbf{X}}_t\in C_{k,\hat Y_t}\qquad \text{a.s.}
\end{align}
Therefore, by the definition of the robust finite-alphabet loss,
\begin{align}
\mathbf{1}\{d_g({\mathbf{X}}_t,\hat {\mathbf{X}}_t)\ge r\}\le\Delta(Y_t,\hat Y_t)\qquad \text{a.s.}
\end{align}
Taking expectations and averaging over $t$ gives, for every $n\ge N_ \varsigma$,
\begin{align}
\frac{1}{n}\sum_{t=1}^{n}\Pr\{d_g({\mathbf{X}}_t,\hat {\mathbf{X}}_t)\ge r\}\le\frac{1}{n}\sum_{t=1}^{n}\mathbb E[\Delta(Y_t,\hat Y_t)]<\epsilon.
\end{align}

The communication rate satisfies, for every $n\ge N_ \varsigma$,
\begin{align}
\frac{1}{n}\mathbb E[\ell(E)]\le I(Y;V)+\delta.
\end{align}
Using the preceding mutual-information identities and inequalities,
\begin{align}
I(Y;V)=I(Y;\hat U)\le I({\mathbf{X}};\hat U)\le R(r,\epsilon)+\frac{ \varsigma}{2}.
\end{align}
Since $\delta< \varsigma/2$,
\begin{align}
\frac{1}{n}\mathbb E[\ell(E)]<R(r,\epsilon)+ \varsigma,\qquad \forall n\ge N_ \varsigma.
\end{align}

For the finitely many blocklengths $n$, choose $E$ to be constant and let the decoder generate $\hat {\mathbf{X}}^n$ independently according to $P_{\mathbf{X}}^{\otimes n}$. Then
\begin{align}
P_{\hat {\mathbf{X}}^n|E}=P_{{\mathbf{X}}^n|E}=P_{\mathbf{X}}^{\otimes n}
\end{align}
also holds for these blocklengths. Hence, the above construction defines a sequence of stochastic block codes satisfying the exact conditional-distribution constraint for every $n$.

Since the rate bound holds for every $n\ge N_ \varsigma$,
\begin{align}
\limsup_{n\to\infty}\frac{1}{n}\mathbb E[\ell(E)]\le R(r,\epsilon)+ \varsigma.
\end{align}
Similarly,
\begin{align}
\limsup_{n\to\infty}\frac{1}{n}\sum_{t=1}^{n}\Pr\{d_g({\mathbf{X}}_t,\hat {\mathbf{X}}_t)\ge r\}\le\epsilon.
\end{align}
Therefore, $R(r,\epsilon)+ \varsigma$ is achievable according to Definition~1. Since $ \varsigma>0$ is arbitrary,
\begin{align}
R_{\mathrm C}(r,\epsilon)\le\inf_{ \varsigma>0}\{R(r,\epsilon)+ \varsigma\}=R(r,\epsilon).
\end{align}

\textbf{Converse:}
Consider any sequence of stochastic block codes achieving rate $R$ for
$(r,\epsilon)$. For each $n$, define
\begin{align}
\epsilon_n:=\frac{1}{n}\sum_{t=1}^{n}
\Pr\{d_g({\mathbf{X}}_t,\hat {\mathbf{X}}_t)\ge r\}.
\end{align}
By the operational distortion constraint, $\limsup_{n\to\infty}\epsilon_n\le\epsilon$.

By the standard time-sharing argument and proof method for converse problem, we have 
\begin{align}
\frac{1}{n}\mathbb E[\ell(E)]\ge I({\mathbf{X}};\hat U_n)\ge R(r,\epsilon_n).
\end{align}

For every $\eta>0$, $\limsup_{n\to\infty}\epsilon_n\le\epsilon$ implies
that $\epsilon_n\le\epsilon+\eta$ for all sufficiently large $n$.
Since $R(r,\epsilon)$ is nonincreasing in $\epsilon$,
\begin{align}
R(r,\epsilon_n)\ge R(r,\epsilon+\eta)
\end{align}
for all sufficiently large $n$. Consequently,
\begin{align}
R\ge\limsup_{n\to\infty}\frac{1}{n}\mathbb E[\ell(E)]\ge R(r,\epsilon+\eta).
\end{align}
Letting $\eta\downarrow0$ and using the continuity of $R(r,\epsilon)$ on $(0,1)$ gives
\begin{align}
R \ge R(r,\epsilon).
\end{align}
Taking the infimum over all achievable rates $R$ yields
\begin{align}
R_{\mathrm C}(r,\epsilon)\ge R(r,\epsilon).
\end{align}
Combining this with the achievability result gives
\begin{align}
R_{\mathrm C}(r,\epsilon)=R(r,\epsilon).
\end{align}

\section{Proof of Proposition~\ref{pro:constraint_covering}\label{app:constraint_covering}}
To handle the limiting case as $\tau \downarrow 0$, we first construct a sequence. Let $\{\tau_n\}_{n\ge1}$ be an arbitrary sequence satisfying
\begin{align}
\lim_{n\rightarrow\infty}{\tau_n}=0.
\end{align}
If
\begin{align}
\liminf_{n\to\infty}\psi_r(\tau_n)=+\infty,
\end{align}
the desired inequality holds immediately. Otherwise, there always exists a subsequence, still denoted by $\{\tau_n\}_{n\ge1}$ for notational simplicity, such that
\begin{align}
\lim_{n\to\infty}\psi_r(\tau_n)=\liminf_{n\to\infty}\psi_r(\tau_n).
\end{align}
For each $n$, by the definition of the infimum in $\psi_r(\tau_n)$, we can choose a nearly optimal $Q_n$ on $\mathcal M$ satisfying
\begin{align}
\int_{\mathcal M}\int_{\mathcal M}\mathbf 1 \left\{d_g(\mathbf x, \hat{\mathbf{x}})\ge r\right\}Q_n(d\mathbf x)Q_n(d \hat{\mathbf{x}})\le\tau_n,
\end{align}
and
\begin{align}
D_{KL}(Q_n\|P_{\mathbf X})\le\psi_r(\tau_n)+\frac{1}{n}.
\end{align}
Thus, $\{Q_n\}_{n\ge1}$ is a sequence of probability distributions. %whose pairwise excess-distortion probabilities approach zero and whose relative entropies approach the corresponding infima.

Since $\mathcal M$ is compact, every sequence of probability distributions on $\mathcal M$ admits a weakly convergent subsequence. Hence, after passing to a further subsequence and relabeling it again as $\{Q_n\}_{n\ge1}$, there exists a probability distribution $Q$ on $\mathcal M$ such that, for every bounded continuous function $f:\mathcal M\to\mathbb R$,
\begin{align}
\lim_{n\to\infty}\int_{\mathcal M}f(\mathbf x)Q_n(d\mathbf x)=\int_{\mathcal M}f(\mathbf x)Q(d\mathbf x).
\end{align}
This weak convergence further implies that, for every bounded continuous function $h:\mathcal M\times\mathcal M\to\mathbb R$,
\begin{align}
\lim_{n\to\infty}\int_{\mathcal M}\int_{\mathcal M}h(\mathbf x, \hat{\mathbf{x}})Q_n(d\mathbf x)Q_n(d \hat{\mathbf{x}})
=\int_{\mathcal M}\int_{\mathcal M}h(\mathbf x, \hat{\mathbf{x}})Q(d\mathbf x)Q(d \hat{\mathbf{x}}).
\end{align}

Now define
\begin{align}
G_r:=\left\{(\mathbf x, \hat{\mathbf{x}})\in\mathcal M\times\mathcal M:d_g(\mathbf x, \hat{\mathbf{x}})>r\right\}.
\end{align}
Since $d_g$ is continuous, $G_r$ is open. Moreover,
\begin{align}
\left\{(\mathbf x, \hat{\mathbf{x}}):d_g(\mathbf x, \hat{\mathbf{x}})>r\right\}\subseteq\left\{(\mathbf x, \hat{\mathbf{x}}):d_g(\mathbf x, \hat{\mathbf{x}})\ge r\right\}.
\end{align}
Therefore, for every $n$,
\begin{align}
\int_{\mathcal M}\int_{\mathcal M}\mathbf 1\left\{d_g(\mathbf x, \hat{\mathbf{x}})>r\right\}Q_n(d\mathbf x)Q_n(d \hat{\mathbf{x}})
\le\int_{\mathcal M}\int_{\mathcal M}\mathbf 1\left\{d_g(\mathbf x, \hat{\mathbf{x}})\ge r\right\}Q_n(d\mathbf x)Q_n(d \hat{\mathbf{x}})\le\tau_n.
\end{align}
Applying the Portmanteau theorem to the open set $G_r$ yields
\begin{align}
\int_{\mathcal M}\int_{\mathcal M}\mathbf 1 \left\{d_g(\mathbf x, \hat{\mathbf{x}})>r\right\}Q(d\mathbf x)Q(d \hat{\mathbf{x}})
\le\liminf_{n\to\infty}\int_{\mathcal M}\int_{\mathcal M}\mathbf 1 \left\{d_g(\mathbf x, \hat{\mathbf{x}})>r\right\}Q_n(d\mathbf x)Q_n(d \hat{\mathbf{x}})
\le\liminf_{n\to\infty}{\tau_n}
=0
\end{align}
Since the integral with respect to $Q$ is nonnegative, we have
\begin{align}\label{eq:Q0}
\int_{\mathcal M}\int_{\mathcal M}\mathbf 1 \left\{d_g(\mathbf x, \hat{\mathbf{x}})>r\right\}Q(d\mathbf x)Q(d \hat{\mathbf{x}})=0.
\end{align}
We claim that the support of $Q$ satisfies
\begin{align}
\operatorname{diam}_g(\operatorname{supp}Q)\le r.
\end{align}
Suppose otherwise. Then there exist $\mathbf x_1,\mathbf x_2\in\operatorname{supp}Q$ such that
\begin{align}
d_g(\mathbf x_1,\mathbf x_2)>r.
\end{align}
By continuity of $d_g$, there exist open neighborhoods $O_1$ and $O_2$ of $\mathbf x_1$ and $\mathbf x_2$, respectively, such that
\begin{align}
d_g(\mathbf{x}_1',\mathbf{x}_2')>r,\qquad \forall(\mathbf{x}_1',\mathbf{x}_2')\in O_1\times O_2.
\end{align}
Since $\mathbf x_1,\mathbf x_2\in\operatorname{supp}Q$, we have
\begin{align}
Q(O_1)>0,\qquad Q(O_2)>0.
\end{align}
Therefore,
\begin{align}
\int_{\mathcal M}\int_{\mathcal M}\mathbf 1 \left\{d_g(\mathbf x, \hat{\mathbf{x}})>r\right\}Q(d\mathbf x)Q(d \hat{\mathbf{x}})\ge Q(O_1)Q(O_2)>0,
\end{align}
which contradicts \eqref{eq:Q0}. Hence,
\begin{align}
\operatorname{diam}_g(\operatorname{supp}Q)\le r.
\end{align}
Therefore, the definition of $\beta_r$ gives
\begin{align}
P_{\mathbf X}(\operatorname{supp}Q)\le\beta_r.
\end{align}

Next, for any probability measures $Q'$ and $P'$, if $Q'(A)=1$, the data-processing inequality for relative entropy under the mapping $\mathbf x\mapsto\mathbf 1\{\mathbf x\in A\}$ yields
\begin{align}
D_{KL}(Q'\|P')\ge-\log_2 P'(A).
\end{align}
Since $Q(\operatorname{supp}Q)=1$, applying this inequality
with $A=\operatorname{supp}Q$ gives
\begin{align}
D_{KL}(Q\|P_{\mathbf X})\ge-\log_2 P_{\mathbf X}(\operatorname{supp}Q)\ge-\log_2 \beta_r.
\end{align}
Finally, relative entropy is lower semicontinuous under
weak convergence, and the chosen near-optimal measures
$Q_n$ satisfy
\begin{align}
\psi_r(\tau_n)\ge D_{KL}(Q_n\|P_{\mathbf X})-\frac{1}{n}.
\end{align}
Consequently,
\begin{align}
\liminf_{n\to\infty}\psi_r(\tau_n)\ge\liminf_{n\to\infty}D_{KL}(Q_n\|P_{\mathbf X})
\ge D_{KL}(Q\|P_{\mathbf X})\ge-\log_2\beta_r.
\end{align}
Since the original sequence $\{\tau_n\}_{n\ge1}$ of positive
numbers converging to zero was arbitrary, we conclude that
\begin{align}
\liminf_{\tau\downarrow0}\psi_r(\tau)\ge-\log_2\beta_r.
\end{align}

\section{Proof of Theorem~\ref{thm_generative_rate_distortion_converse}}
\label{app:generative_rate_distortion_converse}
Fix an arbitrary sufficiently small $\epsilon>0$ and an arbitrary feasible joint distribution. For $P_{\hat U}$-a.e. $u$, define
\begin{align}
\delta_r(u):=\int_{\mathcal M}\int_{\mathcal M}\mathbf 1 \left\{d_g(\mathbf x, \hat{\mathbf{x}})\ge r\right\}Q_u(d\mathbf x)Q_u(d \hat{\mathbf{x}}).
\end{align}
By the Markov condition and conditional-distribution
preservation, $\mathbf X$ and $\hat{\mathbf X}$ are
conditionally independent and identically distributed
according to $Q_u$ given $\hat U=u$.
Consequently, the excess-distortion constraint gives
\begin{align}
\int_{\mathcal U}\delta_r(u)P_{\hat U}(du)\le\epsilon.
\end{align}
For an arbitrary threshold $\tau>\epsilon$, define
\begin{align}
\mathcal G_\tau:=\left\{u\in\mathcal U:\delta_r(u)\le\tau\right\}.
\end{align}
Its complement is
\begin{align}
\mathcal G_\tau^c=\left\{u\in\mathcal U:\delta_r(u)>\tau\right\}.
\end{align}
Since $\delta_r(u)>\tau$ on $\mathcal G_\tau^c$, we have
\begin{align}
\epsilon\ge\int_{\mathcal U}\delta_r(u)P_{\hat U}(du)\ge\int_{\mathcal G_\tau^c}\delta_r(u)P_{\hat U}(du)\ge\tau P_{\hat U}(\mathcal G_\tau^c).
\end{align}
Therefore,
\begin{align}
P_{\hat U}(\mathcal G_\tau)\ge1-\frac{\epsilon}{\tau}.
\end{align}
For every $u\in\mathcal G_\tau$, $Q_u$ satisfies
\begin{align}
\int_{\mathcal M}\int_{\mathcal M}\mathbf 1 \left\{d_g(\mathbf x, \hat{\mathbf{x}})\ge r\right\}Q_u(d\mathbf x)Q_u(d \hat{\mathbf{x}})\le\tau.
\end{align}
Thus, by the definition of $\psi_r(\tau)$,
\begin{align}
D_{KL}(Q_u\|P_{\mathbf X})\ge\psi_r(\tau),\qquad u\in\mathcal G_\tau.
\end{align}
Using the conditional-divergence representation of mutual information and the nonnegativity of relative entropy, we obtain
\begin{align}
I(\mathbf X;\hat U)&=\int_{\mathcal U}D_{KL}(Q_u\|P_{\mathbf X})P_{\hat U}(du)
\ge\int_{\mathcal G_\tau}D_{KL}(Q_u\|P_{\mathbf X})P_{\hat U}(du)
\ge\psi_r(\tau)P_{\hat U}(\mathcal G_\tau)
\ge\left(1-\frac{\epsilon}{\tau}\right)\psi_r(\tau).
\end{align}
Since this inequality holds for every feasible joint distribution, taking the infimum over all feasible choices gives
\begin{align}
R(r,\epsilon)\ge\left(1-\frac{\epsilon}{\tau}\right)\psi_r(\tau).
\end{align}

Choose $\tau=\sqrt{\epsilon}$. Since $\epsilon\in(0,1)$,
we have $\tau>\epsilon$. Moreover, as $\epsilon\downarrow0$,
\begin{align}
\tau\downarrow0,
\qquad
\frac{\epsilon}{\tau}=\sqrt{\epsilon}\downarrow0.
\end{align}
By Proposition~\ref{pro:constraint_covering},
\begin{align}
\liminf_{\epsilon\downarrow0}R(r,\epsilon)
\ge
\liminf_{\epsilon\downarrow0}
(1-\sqrt{\epsilon})\psi_r(\sqrt{\epsilon})
\ge-\log_2\beta_r.
\end{align}
Since $R(r,\epsilon)$ is nonincreasing in $\epsilon$,
its limit as $\epsilon\downarrow0$ exists in the
extended real line. Therefore,
\begin{align}
R(r)
=\lim_{\epsilon\downarrow0}R(r,\epsilon)
\ge-\log_2\beta_r.
\end{align}

\section{Proof of Theorem~\ref{thm_general_alpha_achievability}\label{app:thm_general_alpha_achievability}}
To prove the achievability bound, it suffices to construct a pair of valid probability kernels $P_{\hat U\mid{\mathbf{X}}}$ and $P_{\hat{{\mathbf{X}}}\mid\hat U}$ satisfying the constraints defining $R(r)$ and then compute the corresponding mutual information.

\textbf{Step A: Proof of the infimum}

We first establish the existence of a strictly positive
coverage function with finite logarithmic cost. Since $(\mathcal M,d_g)$ is compact, it is totally bounded. Hence, for every
$r>0$, there exist a positive integer $N\in\mathbb N$ and points
\begin{align}
\mathbf x_1,\ldots,\mathbf x_N\in\mathcal M
\end{align}
such that
\begin{align}
\mathcal M\subseteq\bigcup_{j=1}^N B_{r/2}(\mathbf x_j).
\end{align}
For every $j\in\{1,\ldots,N\}$ and every $\mathbf u,\mathbf v\in B_{r/2}(\mathbf x_j)$, the triangle inequality gives
\begin{align}
d_g(\mathbf u,\mathbf v)\le d_g(\mathbf u,\mathbf x_j)+d_g(\mathbf x_j,\mathbf v)\le r.
\end{align}
Therefore,
\begin{align}
\operatorname{diam}_g(B_{r/2}(\mathbf x_j)) \le r,
\end{align}
so that
\begin{align}
B_{r/2}(\mathbf{x}_j)\in\mathcal A_r, \qquad j=1,\ldots,N.
\end{align}
Thus, $\mathcal A_r$ admits a finite cover of $\mathcal M$.

Set
\begin{align}
A_j:=B_{r/2}(\mathbf x_j),\qquad
q_j:=\frac1N,\qquad j=1,\ldots,N.
\end{align}
The corresponding coverage function is
\begin{align}\label{eq_alpha_zero}
\alpha_0(\mathbf x)
:=\frac1N\sum_{j=1}^N\mathbf1_{A_j}(\mathbf x).
\end{align}
Consequently,
\begin{align}
0\le\mathbb E_{{\mathbf{X}}}\left[-\log_2\alpha_0({\mathbf{X}})\right]\le\log_2 N<\infty.
\label{eq:finite_alpha_exists}
\end{align}
Hence, the class
\begin{align}
\operatorname{VP}_r^s:=\left\{\alpha\in\operatorname{VP}_r:\mathbb E_{{\mathbf{X}}}\left[-\log_2\alpha({\mathbf{X}})\right]<\infty\right\}
\end{align}
is nonempty and for $\alpha\in\operatorname{VP}_r^s$, $\alpha({\mathbf{X}})>0$ almost surely.

Therefore, for each $\alpha \in \operatorname{VP}_r^s$, we construct the auxiliary communication random variable
$\hat{U}$ as 
\begin{align}
\hat{\mathcal U}:=\{1,\ldots,K\}.
\end{align} 
and then the valid probability kernel $P_{\hat{U}\mid{\mathbf{X}}}(k\mid\mathbf x)$ is given by
\begin{align}
P_{\hat{U}\mid{\mathbf{X}}}(k\mid\mathbf x):=\frac{q_k\mathbf 1_{A_k}(\mathbf x)}{\alpha(\mathbf x)},\qquad k\in \hat{\mathcal U}.
\end{align}
If $\alpha(\mathbf{X})=0$, define it arbitrarily.
The probability kernel is valid because for each fixed $k$, the mapping
\begin{align}
\mathbf x \mapsto P_{\hat{U}\mid{\mathbf{X}}}(k\mid\mathbf x)
\end{align}
is measurable, since $A_k$ is measurable and $\alpha$ is measurable and for every $\mathbf x$ satisfying $\alpha(\mathbf x)>0$,
\begin{align}
\sum_{k=1}^KP_{\hat{U}\mid{\mathbf{X}}}(k\mid\mathbf x)=\frac{\sum_{k=1}^Kq_k\mathbf 1_{A_k}(\mathbf x)}{\alpha(\mathbf x)}=1,
\end{align}

For every $k\in\hat{\mathcal U}$, define
\begin{align}
z_k:=\int_{A_k}\frac{1}{\alpha(\mathbf x)}dP_{{\mathbf{X}}}(\mathbf x).
\end{align}
we have 
\begin{align}
z_k\le \frac{P_\mathbf{X}(A_k)}{q_k}\le \frac{1}{q_k}\le \infty.
\end{align}
Therefore, for every $k$ satisfying $z_k>0$, define the generation probability kernel by
\begin{align}
P_{\hat{{\mathbf{X}}}\mid\hat{U}}(E\mid k):=\frac{1}{z_k}\int_{E\cap A_k}\frac{1}{\alpha(\hat{\mathbf x})}dP_{{\mathbf{X}}}(\hat{\mathbf x}),\qquad E\in\mathcal B(\mathcal M).
\end{align}
If $z_k=0$, define $P_{\hat{{\mathbf{X}}}\mid\hat{U}=k}$ arbitrarily. The probability kernel is valid because $\alpha$ and $E\cap A_k$ are Borel measurable, and  for $z_k>0$, 
\begin{align}
P_{\hat{{\mathbf{X}}}\mid\hat{U}}(\mathcal M\mid k)=\frac{1}{z_k}\int_{A_k}\frac{1}{\alpha(\hat{\mathbf x})}dP_{{\mathbf{X}}}(\hat{\mathbf x})=1.
\end{align}

We next verify the conditional distribution-preserving constraint.
The marginal probability of $\hat U=k$ is
\begin{align}
p_k:=P_{\hat U}(k)&=\int_{\mathcal M}P_{\hat U\mid{\mathbf{X}}}(k\mid\mathbf x)dP_{{\mathbf{X}}}(\mathbf x)
=q_k\int_{A_k}\frac{1}{\alpha(\mathbf x)}dP_{{\mathbf{X}}}(\mathbf x)
=q_kz_k.
\end{align}
For every $k$ satisfying $p_k>0$ and every
$E\in\mathcal B(\mathcal M)$, Bayes' rule gives
\begin{align}
P_{{\mathbf{X}}\mid\hat U}(E\mid k)&=\frac{1}{p_k}\int_EP_{\hat U\mid{\mathbf{X}}}(k\mid\mathbf x) dP_{{\mathbf{X}}}(\mathbf x)
=\frac{1}{q_kz_k}\int_{E\cap A_k}\frac{q_k}{\alpha(\mathbf x)}dP_{{\mathbf{X}}}(\mathbf x)
=\frac{1}{z_k}\int_{E\cap A_k}\frac{1}{\alpha(\mathbf x)}dP_{{\mathbf{X}}}(\mathbf x)
=P_{\hat{{\mathbf{X}}}\mid\hat U}(E\mid k).
\end{align}
Therefore,
\begin{align}
P_{\hat{{\mathbf{X}}}\mid\hat U}=P_{{\mathbf{X}}\mid\hat U},\qquad P_{\hat U}\text{-a.s.}
\end{align}

We next verify the distortion constraint. Conditioned on $\hat{U}=k$ with positive probability, the construction of the first kernel gives
\begin{align}
{\mathbf{X}}\in A_k\qquad\text{a.s.},
\end{align}
while the generation kernel gives
\begin{align}
\hat{{\mathbf{X}}}\in A_k\qquad\text{a.s.}
\end{align}
Since
\begin{align}
A_k\in\mathcal A_r,
\end{align}
we have
\begin{align}
\operatorname{diam}_g(A_k)\le r.
\end{align}
Therefore,
\begin{align}
\Pr\left\{d_g({\mathbf{X}},\hat{{\mathbf{X}}})\le r\right\}=1
\end{align}
Since $P_{{\mathbf{X}}}$ admits the density $f=dP_{{\mathbf{X}}}/dV_g$, the constructed conditional distributions are absolutely continuous with respect to $V_g$.
Together with the Markov chain ${\mathbf{X}}-\hat U-\hat{{\mathbf{X}}}$, this implies $P_{{\mathbf{X}},\hat{{\mathbf{X}}}\mid\hat U=k}$ is absolutely continuous with respect to $V_g\otimes V_g$ for every $k$ with $P_{\hat U}(k)>0$.

Also, since the level set
\begin{align}
\left\{(\mathbf \mathbf{x},\hat{\mathbf x}): d_g(\mathbf \mathbf{x},\hat{\mathbf x})=r\right\}
\end{align}
has zero $(V_g\otimes V_g)$-measure for $r>0$, we have
\begin{align}
\Pr\left\{d_g({\mathbf{X}},\hat{{\mathbf{X}}})=r\right\}=0.
\end{align}

Therefore, we have
\begin{align}
\Pr\left\{d_g({\mathbf{X}},\hat{{\mathbf{X}}})\ge r \right\}=0\le \epsilon
\end{align}
Thus, the constructed pair of probability kernels is feasible for the optimization defining $R(r,\epsilon)$ for every $\epsilon\in(0,1)$ and thus feasible for $R(r)$.

It remains to bound the mutual information. For every $P_{\hat{U}\mid{\mathbf{X}}}(k\mid\mathbf x)>0$, $\mathbf 1_{A_k}(\mathbf x)=1$. Therefore, 
\begin{align}
\frac{P_{\hat{U}\mid{\mathbf{X}}}(k\mid{\mathbf{X}})}{P_{\hat{U}}(k)}=\frac{\frac{q_k}{\alpha({\mathbf{X}})}}{q_kz_k}=\frac{1}{\alpha({\mathbf{X}})z_k}.
\end{align}

Hence
\begin{align}
I({\mathbf{X}};\hat{U})&=\sum_{k:p_k>0}\int_{\mathcal M}P_{\hat{U}\mid{\mathbf{X}}}(k\mid\mathbf x)\log_2\frac{P_{\hat{U}\mid{\mathbf{X}}}(k\mid\mathbf x)}{p_k} dP_{{\mathbf{X}}}(\mathbf x)=\sum_{k:p_k>0}\int_{A_k}\frac{q_k}{\alpha(\mathbf x)}\left[-\log_2\alpha(\mathbf x)-\log_2z_k\right] dP_{{\mathbf{X}}}(\mathbf x)\nonumber\\
%&=\mathbb E_{{\mathbf{X}}}\left[-\log_2\alpha({\mathbf{X}})\right]-\sum_{k:p_k>0}p_k\log_2z_k.
\end{align}

For the first term, we have
\begin{align}
&\sum_{k:p_k>0}q_k\int_{A_k}\frac{1}{\alpha(\mathbf x)}\left[-\log_2\alpha(\mathbf x)\right] dP_{{\mathbf{X}}}(\mathbf x)\nonumber\\
&=\int_{\mathcal M}\frac{\sum_{k=1}^Kq_k\mathbf 1_{A_k}(\mathbf x)}{\alpha(\mathbf x)}\left[-\log_2\alpha(\mathbf x)\right] dP_{\mathbf{X}}(\mathbf{x})\nonumber\\
&=\int_{\mathcal M}\left[-\log_2\alpha(\mathbf x)\right] dP_{{\mathbf{X}}}(\mathbf x)\nonumber\\
&=\mathbb E_{{\mathbf{X}}}\left[-\log_2\alpha({\mathbf{X}})\right].
\label{eq:first_MI_term}
\end{align}
Here, indices with $p_k=0$ contribute zero and therefore do not affect the integral.
For the second term, we have
\begin{align}
&\sum_{k:p_k>0}\int_{A_k}\frac{q_k}{\alpha(\mathbf x)}\left[\log_2z_k\right] dP_{{\mathbf{X}}}(\mathbf x)\nonumber\\
&=\sum_{k:p_k>0}
q_k\log_2z_k\int_{A_k}\frac{1}{\alpha(\mathbf x)}dP_{{\mathbf{X}}}(\mathbf x)=\sum_{k:p_k>0}q_kz_k\log_2z_k\nonumber\\
&=\sum_{k:p_k>0}p_k\log_2\frac{p_k}{q_k}=D_{\rm KL}(\mathbf p\Vert\mathbf q)\geq 0 .
\end{align}
where
\begin{align}
\mathbf p:=(p_1,\ldots,p_K),\qquad \mathbf q:=(q_1,\ldots,q_K).
\end{align}

Combining the two terms, for every $\alpha\in\operatorname{VP}_r^s$, we have
\begin{align}
I({\mathbf{X}};\hat{U})=\mathbb E_{{\mathbf{X}}}\left[-\log_2\alpha({\mathbf{X}})\right]-D_{\rm KL}(\mathbf p\Vert\mathbf q).
\label{eq:MI_alpha_z}
\end{align}

Therefore, taking $\epsilon \rightarrow0$ according to the definition of $R(r)$ yields 
\begin{align}\label{eq_Rr_inf}
R(r)\leq \inf_{\alpha\in {\operatorname{VP}}_r^s}\mathbb E_{{\mathbf{X}}}\left[-\log_2\alpha({\mathbf{X}})\right].
\end{align}
By \eqref{eq:finite_alpha_exists}, a coverage function
with finite logarithmic cost exists. Functions with
infinite logarithmic cost therefore do not affect
the infimum. Consequently, we have
\begin{align}
R(r)\leq \inf_{\alpha\in {\operatorname{VP}}_r}\mathbb E_{{\mathbf{X}}}\left[-\log_2\alpha({\mathbf{X}})\right].
\end{align}

\textbf{Step B: Proof of the infimum attained with $L^1(P_{\mathbf X})$-closure}

To obtain a tractable characterization of the above infimum, especially
for symmetric manifolds where the natural optimal coverage function may
arise as a limit of finite mixtures, we extend $\operatorname{VP}_r$ to
its $L^1(P_{{\mathbf{X}}})$-closure $\overline{\operatorname{VP}}_r$.
As $\operatorname{VP}_r\subseteq\overline{\operatorname{VP}}_r$, 
\begin{align}\label{eq:VP_closure_first_direction}
\inf_{\alpha\in \overline{\operatorname{VP}}_r} \mathbb E_{{\mathbf{X}}}\left[-\log_2\alpha({\mathbf{X}})\right]\leq \inf_{\alpha\in \operatorname{\operatorname{VP}}_r} \mathbb E_{{\mathbf{X}}}\left[-\log_2\alpha({\mathbf{X}})\right]
\end{align}
On the other hand, let $\alpha\in\overline{\operatorname{VP}}_r$ satisfy $\mathbb E_{{\mathbf{X}}}\left[-\log_2\alpha({\mathbf{X}})\right]<\infty$, meaning that $\alpha({\mathbf{X}})>0$ is almost surely. By the definition of the $L^1(P_{{\mathbf{X}}})$-closure, there exists a sequence
$\{\alpha_n\}_{n\ge1}\subseteq\operatorname{VP}_r$ such that
\begin{align}
\lim_{n\to\infty}\mathbb E_{{\mathbf{X}}}\left[\left|\alpha_n-\alpha\right|\right]=0.
\end{align}
Combining $\alpha_0$ in \eqref{eq_alpha_zero}, for any fixed $\delta\in(0,1)$, we define
\begin{align}
\beta_{n,\delta}:=(1-\delta)\alpha_n+\delta\alpha_0.
\end{align}
Since $\operatorname{VP}_r$ is convex, we have
\begin{align}
\beta_{n,\delta}\in\operatorname{VP}_r
\end{align}
Define also
\begin{align}
\beta_\delta:=(1-\delta)\alpha+\delta\alpha_0.
\end{align}
Then
\begin{align}
\lim_{n\to\infty}\mathbb E_{{\mathbf{X}}}\left[\left|\beta_{n,\delta}-\beta_\delta\right|\right]=\lim_{n\to\infty}\mathbb E_{{\mathbf{X}}}\left[(1-\delta)\left|\alpha_n-\alpha\right|\right]=0
\end{align}
Moreover, by $\alpha_0\geq\frac{1}{N}$, we have 
\begin{align}
\frac{\delta}{N}\le\beta_{n,\delta}(\mathbf x)\le1,\quad\frac{\delta}{N}\le\beta_{\delta}(\mathbf x)\le1,\quad P_{{\mathbf{X}}}\text{-a.e.}
\label{eq:beta_uniform_positive}
\end{align}
Combining that for fixed $\delta>0$, the function $f(t)=-\log_2 t$ is Lipschitz continuous on $[\delta/N,1]$, with Lipschitz constant $N/(\delta\ln 2)$, we have
\begin{align}
\lim_{n\to\infty}\left|\mathbb{E}_{\mathbf{X}}[-\log_2 \beta_{n,\delta}({\mathbf{X}})]-\mathbb{E}_{\mathbf{X}}[-\log_2 \beta_{\delta}({\mathbf{X}})]\right|
\leq\frac{N}{\delta\ln 2}\lim_{n\to\infty} \mathbb{E}_{\mathbf{X}}\left[\left|\beta_{n,\delta}({\mathbf{X}})-\beta_\delta({\mathbf{X}})\right|\right]=0, \label{eq:F_beta_n_convergence}
\end{align}
% meaning that 
% \begin{align}
%     \lim_{n\to\infty} \beta_{n,\delta}({\mathbf{X}})=\beta_{\delta}({\mathbf{X}})
% \label{eq_limit_betan_to_beta}
% \end{align}

Next, since $-\log_2 t$ is convex, we have
\begin{align}
-\log_2\beta_\delta(\mathbf x)\leq(1-\delta)\left[-\log_2\alpha(\mathbf x)\right]+\delta\left[-\log_2\alpha_0(\mathbf x)\right].
\end{align}
Taking expectations gives
\begin{align}
\mathbb E_{{\mathbf{X}}}\left[-\log_2\beta_\delta({\mathbf{X}})\right]&\leq(1-\delta)\mathbb E_{{\mathbf{X}}}\left[-\log_2\alpha({\mathbf{X}})\right]
+\delta\mathbb E_{{\mathbf{X}}}\left[-\log_2\alpha_0({\mathbf{X}})\right].
\end{align}
Hence,
\begin{align}
\limsup_{\delta\downarrow0}\mathbb E_{{\mathbf{X}}}\left[-\log_2\beta_\delta({\mathbf{X}})\right]\leq\mathbb E_{{\mathbf{X}}}\left[-\log_2\alpha({\mathbf{X}})\right].
\label{eq:beta_delta_limsup}
\end{align}

On the other hand, since
$-\log_2\beta_\delta({\mathbf{X}})\geq0$ and
\begin{align}
\lim_{\delta\downarrow0}\left[-\log_2\beta_\delta({\mathbf{X}})\right]=-\log_2\alpha({\mathbf{X}})\qquad P_{{\mathbf{X}}}\text{-a.s.},
\end{align}
Fatou's lemma gives
\begin{align}
\mathbb E_{{\mathbf{X}}}\left[-\log_2\alpha({\mathbf{X}})\right]\leq\liminf_{\delta\downarrow0}\mathbb E_{{\mathbf{X}}}\left[-\log_2\beta_\delta({\mathbf{X}})\right].
\label{eq:beta_delta_liminf}
\end{align}
Combining \eqref{eq:beta_delta_limsup} and
\eqref{eq:beta_delta_liminf}, we obtain
\begin{align}
\lim_{\delta\downarrow0}\mathbb E_{{\mathbf{X}}}\left[-\log_2\beta_\delta({\mathbf{X}})\right]=\mathbb E_{{\mathbf{X}}}\left[-\log_2\alpha({\mathbf{X}})\right].
\label{eq:beta_delta_cost_limit}
\end{align}

Now, let $\eta>0$ be arbitrary. By
\eqref{eq:beta_delta_cost_limit}, there exists $\delta_0>0$ such that
for every $0<\delta<\delta_0$,
\begin{align}\left|\mathbb E_{{\mathbf{X}}}\left[-\log_2\beta_\delta({\mathbf{X}})\right]-\mathbb E_{{\mathbf{X}}}\left[-\log_2\alpha({\mathbf{X}})\right]\right|<\frac{\eta}{2}.
\end{align}
Fix such a $\delta$. By \eqref{eq:F_beta_n_convergence}, there exists
$n_0\in\mathbb N$ such that for every $n\geq n_0$,
\begin{align}
\left|\mathbb E_{{\mathbf{X}}}\left[-\log_2\beta_{n,\delta}({\mathbf{X}})\right]-\mathbb E_{{\mathbf{X}}}\left[
-\log_2\beta_\delta({\mathbf{X}})\right]\right|<\frac{\eta}{2}.
\end{align}
Therefore,
\begin{align}
\mathbb E_{{\mathbf{X}}}\left[-\log_2\beta_{n,\delta}({\mathbf{X}})\right]<\mathbb E_{{\mathbf{X}}}\left[-\log_2\alpha({\mathbf{X}})\right]+\eta.
\label{eq:recovery_cost}
\end{align}
Since
$\beta_{n,\delta}\in\operatorname{VP}_r$, it follows that
\begin{align}
\inf_{ \zeta\in\operatorname{VP}_r}\mathbb E_{{\mathbf{X}}}\left[-\log_2 \zeta({\mathbf{X}})\right]\leq\mathbb E_{{\mathbf{X}}}\left[-\log_2\alpha({\mathbf{X}})\right]+\eta.
\end{align}
Since $\eta>0$ is arbitrary,
\begin{align}
\inf_{ \zeta\in\operatorname{VP}_r}\mathbb E_{{\mathbf{X}}}\left[-\log_2 \zeta({\mathbf{X}})\right]\leq\mathbb E_{{\mathbf{X}}}\left[-\log_2\alpha({\mathbf{X}})\right].
\end{align}

Since the above inequality holds for every
$\alpha\in\overline{\operatorname{VP}}_r$ with finite logarithmic
moment, while elements with infinite logarithmic moment do not affect
the finite infimum, we obtain
\begin{align}
\inf_{ \zeta\in\operatorname{VP}_r}\mathbb E_{{\mathbf{X}}}\left[-\log_2 \zeta({\mathbf{X}})\right]\leq\inf_{\alpha\in\overline{\operatorname{VP}}_r}\mathbb E_{{\mathbf{X}}}\left[-\log_2\alpha({\mathbf{X}})\right].
\label{eq:VP_closure_reverse_direction}
\end{align}
Combining \eqref{eq:VP_closure_first_direction} and
\eqref{eq:VP_closure_reverse_direction} yields
\begin{align}
\inf_{\alpha\in\operatorname{VP}_r}\mathbb E_{{\mathbf{X}}}\left[-\log_2\alpha({\mathbf{X}})\right]=\inf_{\alpha\in\overline{\operatorname{VP}}_r}\mathbb E_{{\mathbf{X}}}\left[-\log_2\alpha({\mathbf{X}})\right].
\label{eq:VP_VPbar_inf_equality}
\end{align}

It remains to show that the infimum over $\overline{\operatorname{VP}}_r$ is attained. For every $\alpha\in\overline{\operatorname{VP}}_r$, we have
\begin{align}
0\leq\alpha(\mathbf x)\leq 1,\qquad P_{{\mathbf{X}}}\text{-a.e.}
\end{align}
Therefore, $\overline{\operatorname{VP}}_r$ is uniformly integrable in $L^1(P_{{\mathbf{X}}})$. By the Dunford--Pettis theorem \cite{DunfordPettis1940}, $\overline{\operatorname{VP}}_r$ is relatively weakly compact in $L^1(P_{{\mathbf{X}}})$.

Moreover, $\operatorname{VP}_r$ is convex, and hence its $L^1(P_{{\mathbf{X}}})$-closure $\overline{\operatorname{VP}}_r$ is convex and norm closed. Since every norm-closed convex subset of a Banach space is weakly closed, $\overline{\operatorname{VP}}_r$ is weakly compact in $L^1(P_{{\mathbf{X}}})$.

Next, since
\begin{align}
\phi(t):=-\log_2 t,\qquad \phi(0):=+\infty,
\end{align}
is convex and lower semicontinuous on $[0,1]$, the integral functional
\begin{align}
\alpha\mapsto\mathbb E_{{\mathbf{X}}}\left[-\log_2\alpha({\mathbf{X}})\right]
\end{align}
is weakly lower semicontinuous on $L^1(P_{{\mathbf{X}}})$ \cite{Ioffe1977}.

Furthermore, the finite-cover construction provides $\alpha_0\in\operatorname{VP}_r\subseteq\overline{\operatorname{VP}}_r$ such that
\begin{align}
\mathbb E_{{\mathbf{X}}}\left[-\log_2 \alpha_0(\mathbf{X})\right]<\infty.
\end{align}
Hence, the objective functional is finite at least at one element of $\overline{\operatorname{VP}}_r$.

Since $\overline{\operatorname{VP}}_r$ is weakly compact and the objective functional is weakly lower semicontinuous, the infimum is attained. Therefore, there exists $\alpha^\star\in\overline{\operatorname{VP}}_r$ such that
\begin{align}\label{eq_inf_min}
\inf_{\alpha\in\overline{\operatorname{VP}}_r}\mathbb E_{{\mathbf{X}}}\left[-\log_2\alpha({\mathbf{X}})\right]&=\mathbb E_{{\mathbf{X}}}\left[-\log_2\alpha^\star({\mathbf{X}})\right]=\min_{\alpha\in\overline{\operatorname{VP}}_r}\mathbb E_{{\mathbf{X}}}\left[-\log_2\alpha({\mathbf{X}})\right].
\end{align}
Finally, combining \eqref{eq_Rr_inf}, \eqref{eq:VP_VPbar_inf_equality}, and \eqref{eq_inf_min}, we obtain
\begin{align}
R(r)\le\min_{\alpha\in\overline{\operatorname{VP}}_r}\mathbb E_{{\mathbf{X}}}\left[-\log_2\alpha({\mathbf{X}})\right].
\end{align}

%---------------------------------------------%
\footnotesize
\bibliographystyle{IEEEtran}
\bibliography{reference}{}
\end{document}